\documentclass[final,twocolumn]{elsarticle}
                                     
\newcommand*{\stackrelbelow}[2]{ \mathrel{\mathop{#2}\limits_{#1}  }}
\newcommand{\bzer}{\ensuremath{\mathbf{0}}}

\newcommand{\bA}{\ensuremath{\mathbf{A}}}
\newcommand{\bB}{\ensuremath{\mathbf{B}}}
\newcommand{\bC}{\ensuremath{\mathbf{C}}}
\newcommand{\bD}{\ensuremath{\mathbf{D}}}
\newcommand{\bE}{\ensuremath{\mathbf{E}}}
\newcommand{\bF}{\ensuremath{\mathbf{F}}}
\newcommand{\bG}{\ensuremath{\mathbf{G}}}
\newcommand{\bH}{\ensuremath{\mathbf{H}}}
\newcommand{\bI}{\ensuremath{\mathbf{I}}}

\newcommand{\bM}{\ensuremath{\mathbf{M}}}
\newcommand{\bN}{\ensuremath{\mathbf{N}}}

\newcommand{\bP}{\ensuremath{\mathbf{P}}}
\newcommand{\bQ}{\ensuremath{\mathbf{Q}}}
\newcommand{\bR}{\ensuremath{\mathbf{R}}}

\newcommand{\bU}{\ensuremath{\mathbf{U}}}
\newcommand{\bV}{\ensuremath{\mathbf{V}}}
\newcommand{\bW}{\ensuremath{\mathbf{W}}}
\newcommand{\bX}{\ensuremath{\mathbf{X}}}
\newcommand{\bY}{\ensuremath{\mathbf{Y}}}
\newcommand{\bZ}{\ensuremath{\mathbf{Z}}}
\newcommand{\bLamb}{\ensuremath{\mathbf{\Lambda}}}

\newcommand{\bb}{\ensuremath{\mathbf{b}}}
\newcommand{\bc}{\ensuremath{\mathbf{c}}}

\newcommand{\be}{\ensuremath{\mathbf{e}}}
\newcommand{\bff}{\ensuremath{\mathbf{f}}}
\newcommand{\bg}{\ensuremath{\mathbf{g}}}

\newcommand{\bm}{\ensuremath{\mathbf{m}}}
\newcommand{\bn}{\ensuremath{\mathbf{n}}}

\newcommand{\bp}{\ensuremath{\mathbf{p}}}

\newcommand{\bs}{\ensuremath{\mathbf{s}}}

\newcommand{\bu}{\ensuremath{\mathbf{u}}}
\newcommand{\bv}{\ensuremath{\mathbf{v}}}

\newcommand{\by}{\ensuremath{\mathbf{y}}}
\newcommand{\bz}{\ensuremath{\mathbf{z}}}

\newcommand{\bfx}{\ensuremath{\mathbf x}}
\newcommand{\bfA}{\ensuremath{\mathbf A}}

\newcommand{\bfF}{\ensuremath{\mathbf F}}
\newcommand{\bfI}{\ensuremath{\mathbf I}}
\newcommand{\bfb}{\ensuremath{\mathbf b}}
\newcommand{\bfB}{\ensuremath{\mathbf B}}
\newcommand{\bfc}{{\mathbf c}}

\newcommand{\bfC}{{\mathbf C}}

\newcommand{\bfM}{{\mathbf M}}
\newcommand{\bfP}{{\mathbf P}}

\newcommand{\bfV}{{\mathbf V}}
\newcommand{\bfW}{{\mathbf W}}

\newcommand{\bfg}{{\mathbf g}}
\newcommand{\bfG}{{\mathbf G}}
\newcommand{\bfH}{{\mathbf H}}

\newcommand{\bfu}{{\mathbf u}}
\newcommand{\bfy}{{\mathbf y}}

\newcommand{\real}{{\mathsf{Re}}}
\newcommand{\Tr}{{\mathsf{tr}}}
\newcommand{\res}{{\mathsf{res}}}
\newcommand{\Ran}{{\mathsf{Ran}}}

\newcommand{\spn}{{\mathrm{span}}}

\newcommand{\cbfA}{\mbox{\boldmath${\EuScript{A}}$} }
\newcommand{\cbfZ}{\mbox{\boldmath${\EuScript{Z}}$} }
\newcommand{\cbfB}{\mbox{\boldmath${\EuScript{B}}$} }
\newcommand{\cbfC}{\mbox{\boldmath${\EuScript{C}}$} }

\newcommand{\cHtwo}{ { {\mathcal H}_2} }
\newcommand{\cHtwoW}{ {{\mathcal H}_2\mbox{\tiny $(W)$}} }
\newcommand{\cHtwoWmed}{ {{\mathcal H}_2\mbox{\small $(W)$}} }
\newcommand{\cHinf}{ { {\mathcal H}_\infty} }
\newcommand{\cHinfW}{ { {\mathcal H}_\infty\mbox{\tiny $(W)$}} }

\newcommand\IR{ {\mathbb R}}
\newcommand\IC{ {\mathbb C}}

\newtheorem{theorem}{Theorem}
\newtheorem{lemma}[theorem]{Lemma}
\newtheorem{proposition}[theorem]{Proposition}
\newtheorem{corollary}[theorem]{Corollary}

\newenvironment{proof}[1][Proof]{\begin{trivlist}
\item[\hskip \labelsep {\bfseries #1}]}{\end{trivlist}}

\usepackage{savesym}                    
\savesymbol{AND}
\usepackage{algorithmic}
\usepackage{algorithm}                       

\usepackage{euscript,multicol,subfigure}
\usepackage{color}
\usepackage{wrapfig, rotating,calc}
\usepackage{tikz,subfigure}
\usepackage{pgfplots}

\usepackage{amssymb,amsmath}
\usepackage{graphicx}          
                               
\newcounter{subeqnSave}  
\newcounter{parenteqnSave}

\author[kfu]{Tobias Breiten\corref{cor1} \fnref{fn1}}
\ead{tobias.breiten@uni-graz.at}
\author[vt]{Christopher Beattie \fnref{fn2}}
\ead{beattie@vt.edu}
\author[vt]{Serkan Gugercin \fnref{fn2}}
\ead{gugercin@vt.edu}

\cortext[cor1]{Corresponding author}
\address[kfu]{Institute of Mathematics and Scientific Computing, University
of Graz, Heinrichstr. 36, A-8010 Graz, Austria}

\address[vt]{Department of Mathematics, Virginia Tech, Blacksburg, VA
24061-0123, USA}

\fntext[fn1]{Most of this work was completed while this author was with the Max
Planck Institute for Dynamics of Complex Technical Systems, Magdeburg, Germany.}

 \fntext[fn2]{The work of C. Beattie and S. Gugercin was supported in
part by NSF through Grant DMS-1217156.}

\begin{document}

\begin{frontmatter}

\title{Near-optimal frequency-weighted interpolatory model reduction} 



\begin{keyword}                           
frequency-weighting, interpolation, controller reduction, ${\mathcal H}_2$ model
reduction.              
\end{keyword}                             

\begin{abstract}                          
This paper develops an interpolatory framework for weighted-$\cHtwo$ model
reduction of MIMO dynamical systems. 
 A new representation of the weighted-$\cHtwo$ inner products in MIMO settings is
introduced and used to derive associated first-order necessary conditions satisfied by optimal weighted-$\cHtwo$ reduced-order models.  Equivalence of these new interpolatory conditions with earlier Riccati-based conditions given by Halevi is also shown.   
An examination of realizations for equivalent weighted-$\cHtwo$ systems leads then to an algorithm that remains tractable for large state-space dimension.  Several numerical examples
illustrate the effectiveness of this approach and its competitiveness with
Frequency Weighted Balanced Truncation and an earlier interpolatory approach, the Weighted Iterative Rational Krylov Algorithm. 
\end{abstract}

\end{frontmatter}

\section{Introduction} 
\label{sec:intro}

Consider a multiple input/multiple output (MIMO) linear dynamical system having
a  state-space realization (which will be presumed minimal) given by
\begin{equation}
\label{ltisystemintro}
\begin{array}{l}
 \dot{\bfx}(t)  =  \bfA \bfx(t) + \bfB\,\bu(t)\\
\quad \by(t)  = \bfC \bfx(t) + \bD\,\bu(t)
\end{array}
\end{equation}
where $\bfA\in {\mathbb R}^{n\times n}$, $\bfB\in {\mathbb
R}^{n\times m}$,  $\bfC\in {\mathbb R}^{p\times n}$, and 
$\bD\in {\mathbb R}^{p\times m}$  are 
constant matrices.   
 $\bfx(t) \in {\mathbb R}^n$, $\bfu(t) \in {\mathbb R}^m$
 and $\bfy(t) \in {\mathbb R}^p$ are, respectively, 
 the \emph{state},  the \emph{input},  and  the \emph{output} of the system.  
 The \emph{transfer function} of this system is $\bfG(s) =
\bfC(s\bfI-\bfA)^{-1}\bfB + \bD$. 
Following common usage, the underlying system will also be denoted by $\bfG$.
The circumstances of interest for us presume very large state-space dimensions
relative to the input/output dimensions, $n\gg m,p$. 
This leads to fundamental difficulties for any task that involves optimization
or control of this system.  This in turn motivates
\emph{model reduction}:  finding a reduced order model (ROM),
\begin{equation} \label{redsysintro}
\begin{array}{l}
\dot{\bfx}_r (t)  =  \bfA_r \bfx_r (t) + \bfB_r \bfu(t),\\
\quad \bfy_r(t)   =   \bfC_r \bfx_r (t)  + \bD_r\,\bu(t)
\end{array}
\end{equation}
with an associated transfer function $\bfG_r(s) =
\bfC_{r}(s\bfI-\bfA_{r})^{-1}\bfB_{r} + \bD_r$
where $\bfA_r \in \IR^{n_r \times n_r}$, $\bfB_r\in {\mathbb
R}^{n_r\times m}$,   $\bfC_r\in {\mathbb R}^{p\times n_r}$, and
$\bD_r\in {\mathbb R}^{p\times m}$.
The goal is to produce  a greatly reduced state-space dimension, $n_r \ll n$,  yet still
assure that $\bfy_{r}(t)\approx \bfy(t)$ over a large class of inputs $\bfu(t)$.
 This is accomplished by requiring $\bfG_r(s)$ to approximate $\bfG(s)$ very well, in an appropriate sense, 
 which we interpret as 
making $\bfG_r(s)-\bfG(s)$ small with respect to an appropriate system norm. 

For example, one may consider 
approximations that attempt to 
minimize either the $\cHtwo$-error: {\small 
\begin{equation*} 
 \| \bfG-\bfG_r \|_{\cHtwo} \stackrel{\tiny{\mbox{def}}}{=} \left(\frac{1}{2\pi}
\int_{-\infty}^{+\infty}
 \| \bfG(\imath \omega)-\bfG_r(\imath \omega) \|_{F}^2
\,\mathrm{d}\omega\right)^{1/2},
\end{equation*}}
or the $\cHinf$-error: {\small 
\begin{equation*} 
 \| \bfG-\bfG_r \|_{\cHinf} \stackrel{\tiny{\mbox{def}}}{=} \sup_{\omega\in\IR}
 \| \bfG(\imath \omega)-\bfG_r(\imath \omega) \|_{2}.
\end{equation*}}
Here  $\|\bfM\|_F^2=\sum_{i,j}|m_{ij}|^2$ denotes the \emph{Frobenius norm} and $\|\bfM\|_2$ denotes the \emph{spectral norm}
of the matrix $\bfM$. Notice that to ensure that the first error measure is
even \emph{finite},  it is necessary that $\bD_r=\bD$. 

``Typical" inputs, $\bfu(t)$, often will have their power concentrated in known
frequency ranges, and so, some frequency ranges will naturally be more important
than others with regard to ROM fidelity.   This leads in a natural way to
consideration of \emph{weighted}  system errors designed in such a way so as to enhance accuracy in certain frequency ranges while permitting larger errors at other frequencies, and towards that end we 
consider, \emph{weighted} measures of system error such as
\begin{equation*} 
\begin{aligned}
& \| \bfG_r - \bfG \|_{\cHtwoW} \stackrel{\tiny{\mbox{def}}}{=} \| \left(
\bfG_r(s) - \bfG(s) \right)
\bfW(s) \|_{\cHtwo} \\
\mbox{and} & \\
 & \| \bfG_r - \bfG \|_{\cHinfW} \stackrel{\tiny{\mbox{def}}}{=} \| \left(
\bfG_r(s) - \bfG(s) \right)
\bfW(s) \|_{\cHinf}
\end{aligned}
\end{equation*}
{\flushleft
where $\bfW(s)$} is a given input weighting (a ``shaping filter").  
One may specify an output weighting as well, however in the interest of clarity and brevity,
we do not do this here.   We focus on weighted-$\cHtwo$ measures of error so that
for a given system, $\bfG \in \mathbb{\mathcal H}_2$, one seeks a
reduced system 
$\bfG_r \in \mathbb{\mathcal H}_2$ solving: 
{ 
\begin{equation} \label{eq:H2optProb}
\bfG_r =\arg\hspace{-3ex}\min_{\mbox{\tiny $\mathsf{ord}(\tilde{\bfG})\leq
n_r$}} \|\bfG-\tilde{\bfG}\|_{\cHtwoW} 
\end{equation}
}

A variety of shaping filters can be considered.  For example, 
if $\bfW(s)$ were to be chosen to be a transfer function associated with a band-pass filter then
approximation errors at frequencies within the passband would be penalized, while 
approximation error at frequencies lying outside the passband would be discounted. 

Another choice of shaping filter arises from controller reduction:
Consider a linear dynamical system,  $\mathbf{P}$ (the \emph{plant}), with order $n_P$ together with an associated stabilizing controller,  $\bfG$, having order $n$,  
that is connected to $\bP$ in a feedback loop.
Many control design  methodologies,
such as LQG and $\cHinf$ methods, lead ultimately to
controllers whose order is generically as high as the order of the plant, $n\approx n_P$,
see \cite{varga2003aem,ZhDG96} and references therein.  Thus, high-order plants will generally
lead to high-order controllers.  However,
high-order controllers are usually undesirable in real-time applications because this typically translates into unduly complex and costly hardware implementation that may suffer degraded performance both in terms of speed and accuracy. 
Thus, one may prefer to replace $\bfG$ with a reduced order controller,
$\bfG_r$, having order $n_r \ll n$. 

It is often not enough to simply require $\bfG_r$ to be a good approximation to $\bfG$.
In order to accurately recover closed-loop performance, plant dynamics need to be taken into account during the reduction process. This may be achieved through frequency weighting:  Given a  stabilizing controller $\bfG$, if a reduced model,
$\bfG_r$,  has the same number of unstable poles as $\bfG$ and
$$
\left\| [\bfG-\bfG_r]\,\cdot\,\bfP[\bfI + \bfP\bfG]^{-1} \right\|_\cHinf < 1,
$$ 
then,  if $\bfG_r$ is used to replace $\bfG$, 
$\bfG_r$ will also be a stabilizing controller \cite{anderson2002crc,ZhDG96}.  
Seeking $\bfG_r$ to minimize a weighted measure of $\cHtwo$ error as in (\ref{eq:H2optProb}) is an effective proxy, using the weight
$\bfW(s)=\bfP(s)[\bfI + \bfP(s)\bfG(s)]^{-1}$.   This approach has been considered in 
\cite{varga2003aem,anderson2002crc,schelfhout2002ncl,gugercin2004smr,Enn84,wang2002bpp,lin1992mrv,wang1999nfw,sreeram2005fwm} and references therein, leading then to variants of frequency-weighted balanced truncation. 
Related methods in \cite{Hal92,Petersson201432,spanos92} 
are tailored instead towards minimizing a similarly weighted $\mathcal{H}_2$ error, as we do here. 

The main contributions of this paper are threefold. First, we develop a new analysis framework through the introduction of a linear mapping from $\cHtwoWmed$ to $\cHtwo$ that gives a new representation of the weighted-$\cHtwo$ inner product for MIMO systems. This representation allows us to rewrite the weighted-$\cHtwo$ inner product as a regular (unweighted) $\cHtwo$ inner product and leads to interpolatory first-order necessary conditions for optimal weighted-$\cHtwo$ approximation. This  analysis framework allows us to extend the
interpolatory conditions of \cite{Ani13} for the SISO weighted-$\cHtwo$ problem to the MIMO case, and more generally allows us greater flexibility in treating more general settings that involve non-trivial feedthrough terms, which play a crucial role in the weighted-$\cHtwo$ problem.  Second, we show that this new interpolation framework is equivalent to the Riccati-based formulation of Halevi \cite{Hal92}, thus assuring the accuracy of the Riccati-based optimality formulation at a much lower cost.   Finally,
via a detailed examination and a new state-space 
realization for equivalent weighted-$\cHtwo$ systems, we propose a numerical algorithm for weighted-$\cHtwo$ approximation that
remains tractable for large state-space dimension. Unlike the heuristic algorithm introduced in \cite{Ani13}, which is inspired by optimality conditions but does not attempt to satisfy them,  the algorithm proposed here is ``near optimal" in the sense that it directly approximates the weighted optimality conditions and approaches true optimality as reduction order grows.

The rest of the paper is organized as follows: In Section \ref{sec:h2inner}, we introduce the new formulation for the weighted-$\cHtwo$ inner product for MIMO systems based on a bounded linear transformation from $\cHtwoWmed$ to $\cHtwo$ with which we derive interpolatory optimality conditions. The equivalence of these conditions to those of Halevi \cite{Hal92} is proved in Section \ref{sec:halevi} 
followed in Section \ref{sec:nowi} by a description of a numerical algorithm 
 for optimal weighted-$\cHtwo$ approximation based on these conditions. Several numerical examples are given in Section \ref{sec:num};  a summary and conclusions are offered in Section \ref{sec:conc}.

\section{Optimal approximations in a weighted-$\cHtwo$ norm.}  \label{sec:h2inner}
 $\cHinf$ denotes here the set of $m\times m_w$ matrix-valued functions, $\bfW(s)$, having entries, $w_{ij}(s)$, that are analytic for $s$ in the open right half plane and uniformly bounded along the imaginary axis: $\sup_{\omega\in\IR}|w_{ij}(\imath \omega)|$ is finite for all $i,j$. 
A norm may be defined on $\cHinf$ as  $\|\bfW\|_{\cHinf}=\sup_{\omega\in\IR}
\|\bfW(\imath \omega)\|_2$, where $\|\bfM\|_2$ here represents the induced
matrix $2$-norm.  We assume throughout that the weighting functions,  $\bfW(s)$, are
drawn from $\cHinf$. 

For any such weight, $\bfW\in\cHinf$, denote by $\cHtwoWmed$ the set of $p\times m$ matrix-valued functions, $\bfG(s)$, that have components analytic for $s$ in the open right half plane, and such that for each fixed $Re(s)=x>0$, $\bfG(x+\imath y)$ 
is square integrable with respect to $\bfW$ as a function of $y\in(-\infty,\infty)$ in the sense that
$$
\sup_{x>0}\int_{-\infty}^{\infty}\|\bfG(x+\imath y)\bfW(x+\imath y)\|_F^{2}\ dy < \infty.
$$ 
If $\bfG,\, \bfH\in \cHtwoWmed$ are transfer functions representing real 
dynamical systems then an inner product may be defined as
{\small 
\begin{align*} 
&\left\langle \bfG,\ \bfH \right\rangle_{\cHtwoW} \\
&\quad = \frac{1}{2\pi} \int_{-\infty}^{\infty}
\Tr\left( \overline{\bfG(\imath \omega)\bfW(\imath\omega)}
\bfW(\imath\omega)^T\bfH(\imath \omega)^T \right) \, \mathrm{d}\omega \nonumber
\\
&\quad= \frac{1}{2\pi} \int_{-\infty}^{\infty}
\Tr\left( \bfG(-\imath \omega)\bfW(-\imath\omega)\bfW(\imath\omega)^T
\bfH(\imath \omega)^T \right)\, \mathrm{d}\omega.  
\end{align*}
}The associated norm on $\cHtwoWmed$ is 
$$
\|\bfG\|_{\cHtwoW}
=\left( \left\langle \bfG,\ \bfG \right\rangle_{\cHtwoW} \right)^{1/2}.$$

$\cHtwo$ will denote precisely the set $\cHtwoWmed$ with the particular choice $\bfW(s)=\bfI$ (so that$m=m_w$).  
Note that  $\cHtwo\subset \cHtwoWmed$ and for $\bfG,\, \bfH\in \cHtwo$, 
\begin{equation} \label{eq:absH2winner}
\left| \left\langle \bfG,\ \bfH \right\rangle_{\cHtwoW}\right|
\leq  \|\bfW\|_{\cHinf}^2 \  \| \bfG \|_{\cHtwo}\ \| \bfH \|_{\cHtwo}.
\end{equation}

In all that follows, we suppose the weight $\bfW\in\cHinf$ is a rational function with
simple poles at 
 $\{\gamma_{1},\dots,\gamma_{n_w}\}$ and  that it has alternative
representations given by
 \begin{align}
\bfW(s)=\mathbf{C}_w\left(s\mathbf{I}-
\mathbf{A}_w\right)^{-1}\mathbf{B}_w+\bD_w \label{WssRlzn}\\
 \mbox{and}\quad \bfW(s) =\sum_{k=1}^{n_w} \frac{\be_k\, \bff_k^T}{s-\gamma_k} +
\bD_w. \label{WPoleRes}
 \end{align}
 with  $\bfA_w \in \IR^{n_w \times n_w}$, $\bfB_w\in {\mathbb R}^{n_w\times
m_w}$,  $\bfC_w\in {\mathbb R}^{m\times n_w}$, 
 and   $\bD_w\in {\mathbb R}^{m\times m_w}$.   Echoing the setting of 
\cite{Hal92}, our analysis does not require $m=m_w$,  though this may be a natural
choice.  The (matrix-valued) residue of a meromorphic matrix-valued function, $\bfM(s)$,
 at a point $\zeta \in \IC$ will be
denoted as $\res[\bfM(s),\zeta]$, so for example, with $\bfW$ as in (\ref{WPoleRes}),
 $\res[\bfW,\gamma_k]=\be_k\, \bff_k^T$.

Notice that the transfer function, $\bfG$, associated with the system (\ref{ltisystemintro}) will be in $\cHtwoWmed$
if and only if $\bfA$ is stable and $\bD\bD_w = 0.$
For $\bfG\in\cHtwoWmed$,  define
 { \begin{align} 
\mathfrak{F}[\bfG](s)= & \bfG(s)\bfW(s)\bfW(-s)^T \label{Fmap}  \\
&+ \sum^{n_w}_{k=1}
\bfG(-\gamma_{k})\bfW(-\gamma_{k})\frac{\bff_k\,\be_k^T}{s+\gamma_{k}} \nonumber
\end{align}
}
\begin{lemma} \label{Fproperties}
For $\mathfrak{F}$ as defined in (\ref{Fmap})
\begin{enumerate}
\item[a.] $\mathfrak{F}$ is a bounded linear
transformation from $\cHtwoWmed$ to $\cHtwo$.
\item[b.] For any  $\bfG,\,\bfH\in\cHtwo$, 
$\left\langle \bfG,\ \bfH \right\rangle_{\cHtwoW} = \left\langle \mathfrak{F}[\bfG],\ \bfH \right\rangle_{\cHtwo}$.
Hence,  $\mathfrak{F}$ is a positive-definite, selfadjoint linear
operator on $\cHtwo$.
\end{enumerate}
\end{lemma} 
The proof of this lemma and subsequent arguments employ an elementary result
that we list here. It is an immediate corollary to \cite[Lemma 1]{AntBG10}:
\begin{proposition} \label{G1G2H2}
Let $\bfG_1 \in \cHtwo$, $\bfG_2(s) = \frac{\bfc\bfb^T}{s-\mu} \in \cHtwo$, and 
$\bfG_3(s) = \frac{\bfc\bfb^T}{(s-\mu)^2} \in \cHtwo$.
 Then,  
\begin{align*}  
\hspace{-.5em} \langle \bfG_1,\bfG_2 \rangle_{\cHtwo} =  
   \bfc^T\,\overline{\bfG_1}(-\mu),\  &  \left\| \bfG_2 \right\|_{\cHtwo} = \frac{\| \bfc\|\|
\bfb\|}{\sqrt{2 |\real \mu |}} \\
\mbox{and }  
\langle \bfG_1,\bfG_3 \rangle_{\cHtwo}& =- \bfc^T\,\overline{\bfG_1}^{\,\prime}\!(-\mu)\bfb. 
  \end{align*} 
\end{proposition} 

\begin{proof}\textbf{of Lemma \ref{Fproperties}}:
Clearly, $\mathfrak{F}[\bfG]$ is linear in $\bfG$. 
 Let  $\bfG\in \cHtwoWmed$. 
 $\bfG(s)\bfW(s)\bfW(-s)^T$ has simple poles in the right half plane
at 
$-\gamma_{1},-\gamma_{2},\dots,-\gamma_{n_w},$ and 
{\small
\begin{align*}
\res[&\bfG(s)\bfW(s)  \bfW(-s)^T,-\gamma_{k}]\\
&=\lim_{s\rightarrow -\gamma_{k} }(s+\gamma_{k})\bfG(s)\bfW(s)\bfW(-s)^T \\
 &=\bfG(-\gamma_{k})\bfW(-\gamma_{k})\lim_{s\rightarrow -\gamma_{k}
}(s+\gamma_{k})\bfW(-s)^T  \\
  &=  -\bfG(-\gamma_{k})\bfW(-\gamma_{k})\lim_{s\rightarrow \gamma_{k}
}(s-\gamma_{k})\bfW(s)^T \\
  & =-\bfG(-\gamma_{k})\bfW(-\gamma_{k})\cdot\res[\bfW(s)^T,\gamma_{k}] \\
 &=-\bfG(-\gamma_{k})\bfW(-\gamma_{k})\ \bff_k\ \be_k^T.
\end{align*}
}Thus $\mathfrak{F}[\bfG](s)$ is analytic in the right-half plane.  To show that
$\mathfrak{F}[\bfG]\in\cHtwo$, observe first that $\bfG\cdot\bfW\in \cHtwo$ so that
for each  $k=1,\ldots,n_w$ :
\begin{align*}
\|\bfG&(-\gamma_k)\bfW(-\gamma_k)\|_2=\max_{\bfu,\,\bv}\frac{\bfu^*\,\left[
\bfG(-\gamma_k)\bfW(-\gamma_k)\right]\,\bv}{\|\bfu\|_2\, \|\bv\|_2} \\
&= \max_{\bfu, \bv}\frac{1}{\|\bfu\|_2\,
\|\bv\|_2}\left\langle\bfG(s)\bfW(s),\,\frac{\bv\bfu^*}{s-\gamma_k}\right\rangle_{\cHtwo}
\\
&\leq  \|\bfG\bfW\|_{\cHtwo} \cdot \max_{\bfu, \bv}\frac{\quad \left\|\frac{\bv
\bfu^*}{s-\gamma_k}\right\|_{\cHtwo}}{\|\bfu\|\,\|\bv\|} 
= \frac{\|\bfG\|_{\cHtwoW}}{\sqrt{2\, |\real\,\gamma_k |}},
\end{align*}
where the inequality follows from the Cauchy-Schwarz inequality in $\cHtwo$ and the final equality follows from Proposition \ref{G1G2H2}.
Notice that this amounts to the observation that point evaluation in the right half-plane
is a continuous map from $\cHtwoWmed$ to $\IC^{m\times p}$.  We now use this to
calculate
\begin{align*}
\|\mathfrak{F}&[\bfG]\|_{\cHtwo} \leq  \|\bfW\|_{\cHinf}\,
\|\bfG(s)\bfW(s)\|_{\cHtwo} \\
&+\sum^{n_w}_{k=1}
\|
\bfG(-\gamma_{k})\bfW(-\gamma_{k})\frac{\bff_k\be_k^T}{s+\gamma_{k}}\|_{\cHtwo}
\\
&\leq  \left( \|\bfW\|_{\cHinf}
+ \sum_{k=1}^{n_w} \frac{\|\bff_k\|\,\|\be_k\|}{\sqrt{2\, |\real\,\gamma_k|}}
\right)\  \|\bfG\|_{\cHtwoW},
\end{align*}
where we have used the triangle inequality in $\cHtwo$ and the observation that 
$\|\bM\bN\|_F\leq \|\bM\|_2 \|\bN\|_F$ for conforming matrices $\bM$ and $\bN$. Thus,
$\mathfrak{F}$ is a bounded linear transformation from $\cHtwoWmed$ to
$\cHtwo$.  

For assertion \ref{Fproperties}b, suppose first that
$\bfH$ has simple poles  $\{\mu_1,\dots,\mu_{\ell}\}$.  Note that
since $\mathfrak{F}[\bfG](-s)$ is analytic in the left half plane, 
$\mathfrak{F}[\bfG](-s)\bfH(s)^T$ will have poles  in the left halfplane
exactly at  $\{\mu_1,\dots,\mu_{\ell}\}$.

For any  $R>0$,   define a semicircular contour in the left
halfplane:
$
\mathcal{C}_{R}=\left\{z\left| z=\imath\omega \mbox{ with
}\omega\in[-R,R]\right. \right\}\cup \left\{z\left| z=R\,
e^{\imath\theta} \mbox{ with
}\theta\in[\frac{\pi}{2},\frac{3\pi}{2}]\right. \right\}.
$
For $R$ large enough, the region bounded by $\mathcal{C}_{R}$ contains 
$ \{\mu_1,\dots,\mu_{\ell}\} $.  Using the Residue Theorem and linearity of the
trace, we find 
{\small
 \begin{align*}
&\left\langle \mathfrak{F}[\bfG],\ \bfH \right\rangle_{\cHtwo} 
= \frac{1}{2\pi} \int_{-\infty}^{+\infty}\Tr\left(\mathfrak{F}[\bfG](-\imath
\omega)\ \bfH(\imath \omega)^T\right) \,\mathrm{d}\omega\\
&= \lim_{R\rightarrow\infty} \frac{1}{2\pi\imath}
\int_{\mathcal{C}_{R}}\Tr\left(\mathfrak{F}[\bfG](-s)\ \bfH(s)^T\right)
\,\mathrm{d}\omega\\
 & = \sum_{k=1}^{\ell} \Tr\left(\mbox{\textsf{res}}[\mathfrak{F}[\bfG](-s)
\bfH(s)^T,\mu_k]\right) \\
 & = \sum_{k=1}^{\ell} \Tr\left(\mathfrak{F}[\bfG](-\mu_k)
\mbox{\textsf{res}}[\bfH,\mu_k]^T\right) \\
 & = \sum_{k=1}^{\ell} \Tr\left(
\bfG(-\mu_k)\bfW(-\mu_k)\bfW(\mu_k)^T\mbox{\textsf{res}}[\bfH,\mu_k]^T\right) \\
& \quad + \sum_{k=1}^{\ell} \sum^{n_w}_{i=1}  \Tr\left(
\bfG(-\gamma_{i})\bfW(-\gamma_{i})\frac{\bff_i\be_i^T}{-\mu_k+\gamma_{i}}\mbox{
\textsf{res}}[\bfH,\mu_k]^T\right) \\
& = \sum_{k=1}^{\ell} \Tr\left(
\bfG(-\mu_k)\bfW(-\mu_k)\bfW(\mu_k)^T\mbox{\textsf{res}}[\bfH,\mu_k]^T\right) \\
& \quad +  \sum^{n_w}_{i=1}  \Tr\left(
\bfG(-\gamma_{i})\bfW(-\gamma_{i})\bff_i\be_i^T
\sum_{k=1}^{\ell}\frac{\mbox{\textsf{res}}[\bfH,\mu_k]^T}{\gamma_{i}-\mu_k}
\right) 
\end{align*}
} Since $\bfH$ has simple poles and is in $\cHtwo$,
$\sum_{k=1}^{\ell}\frac{\mbox{\textsf{res}}[\bfH,\mu_k]^T}{s-\mu_k}= \bfH(s)^T$.
Note that $ \{\mu_1,\dots,\mu_{\ell}\}\cup \{\gamma_{1},\dots,\gamma_{n_w}\}$ is
 precisely the set
of poles in the left half plane for the meromorphic function
$\bfG(-s)\bfW(-s)\bfW(s)^T \bfH(s)^T$.

So, we continue:
{\small
 \begin{align*}
&\left\langle \mathfrak{F}[\bfG],\ \bfH \right\rangle_{\cHtwo} \\
&\quad = \sum_{k=1}^{\ell} \Tr\left(
\bfG(-\mu_k)\bfW(-\mu_k)\bfW(\mu_k)^T\mbox{\textsf{res}}[\bfH,\mu_k]^T\right) \\
& \qquad +  \sum^{n_w}_{i=1}  \Tr\left(
\bfG(-\gamma_{i})\bfW(-\gamma_{i})\mbox{\textsf{res}}[\bfW,\gamma_i]
^T\bfH(\gamma_{i})^T\right)\\
&\quad =  \lim_{R\rightarrow\infty} \frac{1}{2\pi\imath} \int_{\mathcal{C}_{R}}
\Tr\left(\bfG(-s)\bfW(-s)\bfW(s)^T \bfH(s)^T\right)\  \mathrm{d}s \\
&\quad= \frac{1}{2\pi} \int_{-\infty}^{+\infty}\Tr\left(\bfG(-\imath \omega)\
\bfW(-\imath \omega)\bfW(\imath \omega)^T\bfH(\imath \omega)^T\right)
\,\mathrm{d}\omega\\
&\quad=\left\langle \bfG,\ \bfH \right\rangle_{\cHtwoW} 
\end{align*}
}  
This remains true independent of whether $\bfH$ has
simple poles or not: Take a sequence, $\bfH_k$, converging to $\bfH$ in
$\cHtwo$ with each $\bfH_k$ having simple poles. Then, appeal to the
continuity of the expressions
$\left\langle \bfG,\ \bfH_k \right\rangle_{\cHtwoW} 
= \left\langle \mathfrak{F}[\bfG],\ \bfH_k \right\rangle_{\cHtwo}$ with respect
to the $\cHtwo$ norm.  

$\mathfrak{F}$ is positive-definite and selfadjoint on $\cHtwo$ because, for $\bfG,\,\bfH\in\cHtwo$,
 {\small
\begin{align*}
\left\langle \mathfrak{F}[\bfG],\ \bfH \right\rangle_{\cHtwo}&= \left\langle
\bfG,\ \bfH \right\rangle_{\cHtwoW} 
=\overline{\left\langle \bfH,\ \bfG \right\rangle_{\cHtwoW}}  \\
&=\overline{\left\langle \mathfrak{F}[\bfH],\ \bfG \right\rangle_{\cHtwo}}
=\left\langle \bfG,\  \mathfrak{F}[\bfH] \right\rangle_{\cHtwo}
\end{align*}
}and 
$\left\langle \mathfrak{F}[\bfG],\ \bfG \right\rangle_{\cHtwo}=
\left\langle \bfG,\ \bfG \right\rangle_{\cHtwoW} > 0  \text{\, \, if  \, \,} 
\bfG\neq 0.  ~~\Box$
\end{proof}
Given state-space realizations for $\bfW\in\cHinf$ and $\bfG\in \cHtwoWmed$, one may obtain an explicit state-space realization for $\mathfrak{F}[\bfG](s)$.
\begin{lemma}\label{lem:conn_lcs_F}
  Suppose $\bW \in {\mathcal H}_{\infty}$ has simple poles at
$\{\gamma_1,\dots,\gamma_p\}$ and $\bG \in \cHtwoWmed $. Suppose further that $\bfW(s)$ has a realization as given in (\ref{WssRlzn}) and $\bfG(s) =
\bfC(s\bfI-\bfA)^{-1}\bfB + \bD$ from (\ref{ltisystemintro}).

 Then $\mathfrak{F}[\bG](s)$ as defined in (\ref{Fmap}) has a realization given
by 
{\small 
\begin{align}\label{eq:trans_func_repr}
  \mathfrak{F}[\bG]&(s) = \cbfC_\mathfrak{F} (s \bfI - \cbfA_\mathfrak{F})^{-1} \cbfB_\mathfrak{F}\\
 =&\underbrace{\begin{bmatrix} \bC & \bD\bC_w \end{bmatrix}}_{\cbfC_{\mathfrak{F}}}
 \Bigg( s\bfI - 
 \underbrace{\begin{bmatrix} \bA & \bB\bC_w \\ \bzer & \bA_w
\end{bmatrix}}_{{\cbfA_\mathfrak{F}}}\Bigg)^{-1}
\underbrace{ \begin{bmatrix}\bZ\bC_w^T + \bB\bD_w\bD_w^T \\ \bP_w \bC_w^T +
\bB_w\bD_w^T \end{bmatrix}}_{\cbfB_{\mathfrak{F}}}, \nonumber
\end{align}}
where $\bP_w$ and $\bZ$  solve, respectively, 
\begin{align}
  \bA_w & \bP_w + \bP_w \bA_w^T + \bB_w\bB_w^T = \bzer \quad \mbox{and} \label{eq:weight_lyap} \\
  \bA   \bZ& + \bZ \bA_w^T + \bB (\bC_w\bP_w+\bD_w\bB_w^T) =\bzer. \label{eq:weight_sylv}
\end{align}
\end{lemma}
\begin{proof}
 We evaluate (\ref{eq:trans_func_repr}) in two parts.  Note first that since
$\bG \in \cHtwoWmed$, $\bD\bD_w=\bzer$.  We may directly compute a realization of
$\bG(s)\cdot\bW(s)$:
{\small 
\begin{align}
   &\begin{bmatrix} \bC & \bD\bC_w \end{bmatrix} 
   \begin{bmatrix} s\bI - \bA & -\bB\bC_w \\ \bzer & s\bI-\bA_w \end{bmatrix}^{-1}
   \begin{bmatrix} \bB\bD_w \\  \bB_w \end{bmatrix} \nonumber \\
          &\quad = \begin{bmatrix} \bC & \bD\bC_w \end{bmatrix} 
   \begin{bmatrix} (s\bI - \bA)^{-1}\bB \bW(s) \\ (s\bI-\bA_w)^{-1}\bB_w
\end{bmatrix}   = \bG(s)\bW(s). \label{eq:indp_trans_func}
 \end{align}}
$\bA_w$ has distinct eigenvalues by hypothesis; let its eigenvalue
decomposition be given 
as $\bA_w = \bU \mathbf{\Gamma} \bU^{-1},$ with $\mathbf{\Gamma} =
\mathrm{diag}(\gamma_1,\dots,\gamma_{n_w}).$ Postmultiply
\eqref{eq:weight_lyap} with $\bU^{-T}$:  
\begin{align*}
  \bA_w \tilde{\bP}_w + \tilde{\bP}_w \mathbf{\Gamma} + \bB_w \tilde{\bF} =
\bzer,
\end{align*} 
where { $\bP_w\bU^{-T}=\tilde{\bP}_w=\left[ \tilde{\bp}_1,\,\tilde{\bp}_2,\,\ldots,\,\tilde{\bp}_{n_w}  \right]$}
 and { $\bB_w^T\bU^{-T}=\tilde{\bF} = \left[ \tilde{\bff}_1,\,\tilde{\bff}_2,\,\ldots,\,\tilde{\bff}_{n_w}  \right].$ }
 Since $\Gamma$ is a diagonal matrix, we may solve for each
column of $\tilde{\bP}_w$ independently: 
  $\tilde{\bp}_k = (-\gamma_k \bI-\bA_w)^{-1}\bB_w \tilde{\bff}_k$.
Then defining $\tilde{\bE}=\bC_w \bU =\left[ \tilde{\be}_1,\,\tilde{\be}_2,\,\ldots,\,\tilde{\be}_{n_w}  \right]$, we have
\begin{align*} 
 \bP_w\bC_w^T &= \bP_w \bU^{-T} \bU^T\bC_w^T= \tilde{\bP}_w
\tilde{\bE}^T\\&=\sum_{k=1}^{n_w} (-\gamma_k
\bI-\bA_w)^{-1}\bB_w \tilde{\bff}_k\tilde{\be}_k^T.
\end{align*}
We follow the same development for  \eqref{eq:weight_sylv}; postmultiplication with  $\bU^{-T}$ yields
$$\bA \tilde{\bZ} + \tilde{\bZ}\mathbf{\Gamma}+\bB(\bC_w \tilde{\bP}_w +\bD_w
\tilde{\bF})=\bzer,$$
where $\tilde{\bZ}=\bZ\,\bU^{-T}=\left[ \tilde{\bz}_1,\,\tilde{\bz}_2,\,\ldots,\,\tilde{\bz}_{n_w}  \right]$.
Note  that 
$$
\bC_w\tilde{\bp}_{k} + \bD_w \tilde{\bff}_k = \bW(-\gamma_k)\tilde{\bff}_k
$$
so that
$
 \tilde{\bz}_k = (-\gamma_k \bI-\bA)^{-1}\bB\bW(-\gamma_k)\tilde{\bff}_k.
$
Drawing all together, we obtain
\begin{align*}
\bZ\bC_w^T &=\bZ\bU^{-T} \bU^T \bC_w^T = \tilde{\bZ} \tilde{\bE}^T \\ &=
\sum_{k=1}^{n_w} (-\gamma_k\bI-\bA)^{-1}\bB \bW(-\gamma_k)\tilde{\bff}_k\tilde{\be}_k^T.
\end{align*}
With these expressions, the remaining contribution to \eqref{eq:trans_func_repr} becomes
{\small
\begin{align*}
  &\begin{bmatrix} \bC &\bD \bC_w\end{bmatrix}
 \begin{bmatrix}s\bI-\bA & -\bB\bC_w \\ \bzer & s\bI-\bA_w\end{bmatrix}^{-1}
 \begin{bmatrix} \bZ \bC_w^T \\  \bP_w \bC_w^T\end{bmatrix}\\[.15in]
& \quad= \bC (s \bI-\bA)^{-1}\bZ \bC_w^T + \bG(s) \bC_w (s\bI-\bA_w)^{-1}\bP_w \bC_w^T \\[.15in]
&\quad= \sum_{k=1}^{n_w} \bC(s \bI-\bA)^{-1}(-\gamma_k\bI-\bA)^{-1}\bB\bW(-\gamma_k) \tilde{\bff}_k \tilde{\be}_k^T \\
&  \qquad + \sum_{k=1}^{n_w}
\bG(s)\bC_w(s\bI-\bA_w)^{-1}(-\gamma_k\bI-\bA_w)^{-1}\bB_w \tilde{\bff}_k \tilde{\be}_k^T \\
\end{align*}
}The following easily verified resolvent identity allows further simplification:
{\small 
\begin{equation}\label{eq:resolvent_ident}
\begin{aligned}
&\left(s\mathbf{I}-\mathbf{A}\right)^{-1}\left(-\gamma_k\mathbf{I}-\mathbf{A}
\right)^{-1} \\ &\qquad =
\frac{1}{s+\gamma_k}\left(-\gamma_k\mathbf{I}-\mathbf{A}\right)^{-1}-\frac{1}{
s+\gamma_k}\left(s\mathbf{I}-\mathbf{A}\right)^{-1}.
\end{aligned}
\end{equation}
}
Which then yields,
{\small 
\begin{align*}
\ldots & = \sum_{k=1}^{n_w} \frac{1}{s+\gamma_k}\left(\bG(-\gamma_k)-\bG(s)\right)\bW(-\gamma_k)\tilde{\bff}_k\tilde{\be}_k^T \\
& \qquad\qquad + \sum_{k=1}^{n_w} \frac{1}{s+\gamma_k}\bG(s)\left(\bW(-\gamma_k)-\bW(s)\right)\tilde{\bff}_k\tilde{\be}_k^T  \\
&= \sum_{k=1}^{n_w} 
\bG(-\gamma_k)\bW(-\gamma_k)\frac{\tilde{\bff}_k\tilde{\be}_k^T}{s+\gamma_k} - \bG(s)\bW(s)
\sum_{k=1}^{n_w} \frac{\tilde{\bff}_k\tilde{\be}_k^T}{s+\gamma_k}
\end{align*}
Postmultiplying \eqref{eq:indp_trans_func} with $\bD_w^T$ and combining with this last expression  gives
\begin{align*}
 \begin{bmatrix} \bC &\bD \bC_w\end{bmatrix} &  
 \begin{bmatrix}s\bI-\bA & -\bB\bC_w \\ \bzer & s\bI-\bA_w\end{bmatrix}^{-1}
\begin{bmatrix}\bZ\bC_w^T + \bB\bD_w\bD_w^T \\ \bP_w \bC_w^T + \bB_w\bD_w^T
\end{bmatrix}\\
 &=\bG(s)\bW(s)\left(\sum_{k=1}^{n_w}
\frac{\tilde{\bff}_k\tilde{\be}_k^T}{-s-\gamma_k} +\bD_w^T\right) \\ & \qquad +
\sum_{k=1}^{n_w} 
\bG(-\gamma_k)\bW(-\gamma_k)\frac{\tilde{\bff}_k\tilde{\be}_k^T}{s+\gamma_k} 
\\
  &= \mathfrak{F}[\bG](s). \quad \Box
\end{align*}
} 
 \end{proof}

\begin{lemma} \label{SylEqnSoln}
Suppose $\mathbf{M}_1$ and $\mathbf{M}_2$ are stable matrices.  The unique solution, $\mathbb{X}$, 
to the Sylvester equation
$$
\mathbf{M}_1\mathbb{X}+ \mathbb{X}\mathbf{M}_2 + \mathbf{N}=\bzer,
$$

is given by
$$
\mathbb{X}=\frac{1}{2\pi}\int_{-\infty}^{+\infty} (-\imath\omega\bfI
-\mathbf{M}_1)^{-1}  \mathbf{N} (\imath\omega\bfI
-\mathbf{M}_2)^{-1}\,\mathrm{d}\omega
$$
\end{lemma} 

\begin{lemma} \label{Fextension}
For $\mathfrak{F}$ as defined in (\ref{Fmap})
and any  $\bfG,\,\bfH\in\cHtwoWmed$, let $\bfH = \bC_H(s\bI -
\bA_H)^{-1}\bB_H + \bD_H.$  Then, 
\begin{enumerate}
\item[a.] $ \left\langle \mathfrak{F}[\bfG],\, \bD_H \right\rangle_{\cHtwo} 
 = \frac{1}{2}\, \left\langle \bfG,\ \bD_H \right\rangle_{\cHtwoW}$
\item[b.] $ \left\langle \mathfrak{F}[\bfG],\, \bfH \right\rangle_{\cHtwo}=
\left\langle \bfG,\, \bfH \right\rangle_{\cHtwoW} -\frac{1}{2} \left\langle \bfG,\ \bD_H \right\rangle_{\cHtwoW}$
\end{enumerate}
\end{lemma} 
\begin{proof}
We may decompose $\bfH$ as $\bfH(s)=\bfH_0(s)+\bD_H$ with $\bfH_0\in\cHtwo$.  Since $\bfG,\,\bfH\in\cHtwoWmed$, 
$\bD_H\cdot\bD_w=\bzer$ and $\bD\cdot\bD_w=\bzer$. 
Using the realization of $\bfG\bfW$ in \eqref{eq:indp_trans_func}, we calculate
\begin{align*}
&\left\langle \bfG,\ \bD_H \right\rangle_{\cHtwoW} = \left\langle \bfG\bfW,\
\bD_H\bfW \right\rangle_{\cHtwo}  \\
& \quad = \frac{1}{2\pi}
\int_{-\infty}^{+\infty}\Tr\left(\bfG(-\imath\omega)\bfW(-\imath\omega)\
\bfW(\imath \omega)^T\bD_H^T\right) \,\mathrm{d}\omega\\
& \quad = \Tr\left(\begin{bmatrix} \bC & \bD\bC_w \end{bmatrix} \mathbb{X}
  \bC_w^T\bD_H^T\right)
\end{align*}
where 
{\small
\begin{align*}
\mathbb{X}&=\frac{1}{2\pi} \int_{-\infty}^{+\infty} 
(-\imath \omega \bI -\cbfA_{\mathfrak{F}})^{-1}
    \begin{bmatrix} \bB\bD_w \\  \bB_w \end{bmatrix} 
\bB_w^T(\imath\omega\bI-\bA_w^T)^{-1} \,\mathrm{d}\omega
\end{align*}}From Lemma \ref{SylEqnSoln}, this $\mathbb{X}$  is the unique
solution to the Sylvester equation

$$
\cbfA_{\mathfrak{F}}
\mathbb{X} + 
\mathbb{X}\ \bA_w^T +  \begin{bmatrix} \bB\bD_w\bB_w^T \\  \bB_w\bB_w^T \end{bmatrix}  =\bzer.
$$
Recalling (\ref{eq:weight_lyap}) and (\ref{eq:weight_sylv}), $\mathbb{X}$ evidently may be expressed as 
$
\mathbb{X}= \left[\begin{array}{l} \bZ \\  \bP_w  \end{array}\right].
$
Thus,  $\left\langle \bfG,\ \bD_H \right\rangle_{\cHtwoW} = 
\Tr\left( \bC\bZ\bC_w^T\bD_H^T + \bD\bC_w\bP_w\bC_w^T\bD_H^T \right) $.

Conversely, we may use (\ref{eq:trans_func_repr}), take account that
$\bD_w^T\bD_H^T=\bzer$, and calculate:
{\footnotesize
\begin{align*}
 &\left\langle \mathfrak{F}[\bfG],\ \bD_H \right\rangle_{\cHtwo} = 
 \frac{1}{2\pi} \int_{-\infty}^{+\infty}\Tr\left( 
 \cbfC_{\mathfrak{F}}
  (-\imath \omega \bI-\cbfA_{\mathfrak{F}})^{-1}
\cbfB_{\mathfrak{F}}\bD_H^T
\right)\,\mathrm{d}\omega\\
 &= \Tr \left(
 \cbfC_{\mathfrak{F}}
\left(\frac{1}{2\pi} \int_{-\infty}^{+\infty}
(-\imath \omega \bI-\cbfA_{\mathfrak{F}})^{-1}
\,\mathrm{d}\omega\right)
 \begin{bmatrix}\bZ\bC_w^T \\ \bP_w \bC_w^T \end{bmatrix}\bD_H^T  \right),
\end{align*}
}where the integral limit is to be interpreted as a principal value.
Because the matrix 
$\cbfA_{\mathfrak{F}}$
is stable, the integral reduces to $\pi\bfI$, so we have:
{\small 
\begin{align*}
 \left\langle \mathfrak{F}[\bfG],\ \bD_H \right\rangle_{\cHtwo} 
 &= \frac{1}{2}\, \Tr \left(  \bC\bZ\bC_w^T\bD_H^T + \bD\bC_w\bP_w\bC_w^T\bD_H^T
\right)  \\
 &= \frac{1}{2}\, \left\langle \bfG,\ \bD_H \right\rangle_{\cHtwoW}
\end{align*}
}
Part (b) is shown similarly.  We omit details. $\quad \Box$
\end{proof}

\subsection{Interpolatory  weighted-$\cHtwo$ optimality conditions}
\label{sec:weightedh2cond}

The feasible set for (\ref{eq:H2optProb}) consists of all stable transfer functions
in $\cHtwoWmed$ having order $n_r$ or less. This is a nonconvex set, hence as a
practical matter, finding a global minimizer is extremely difficult. Instead,
one typically seeks efficient \emph{local} minimizers.
Methods proposed in  \cite{Hal92} and  \cite{spanos92} may be used to
 find local minimizers to (\ref{eq:H2optProb}). However, these methods require solving a sequence of large-scale Lyapunov or
Riccati equations and so, rapidly become computationally intractable as system
order, $n$, and shaping filter order, $n_w$, increase. 

We approach (\ref{eq:H2optProb}) instead within an interpolatory framework similar
to that developed in \cite{Ani13} for SISO systems.  Computational complexity for interpolatory methods grows more slowly with increasing $n$ and $n_w$, hence much larger problems are feasible.   In contrast to the (SISO) results of \cite{Ani13},  we are able to treat general MIMO settings including non-zero feedthrough terms, which proves essential for weighted-$\cHtwo$ approximation.
  The  algorithm derived in \cite{Ani13} is heuristic, to the extent that it is inspired by necessary (SISO) optimality conditions but does not seek directly to satisfy them.  Our new algorithm proposed in Section \ref{sec:nowi}, on the other hand,
 directly originates from newly derived MIMO necessary conditions and  uses significantly different model reduction spaces,
ultimately producing near-optimal reduced models that will approach true optimality as reduction order $n_r$ grows.

We first derive interpolatory conditions that necessarily must hold for any reduced system,
$\bfG_r$, that solves (\ref{eq:H2optProb}).

\begin{theorem} \label{thm:weight_intp_cond}
Suppose that $\bfG_r\in \cHtwoWmed$ is a solution to (\ref{eq:H2optProb}). 
Suppose further that $\bfG_r$ has only simple poles,
$\{\lambda_1,\,\ldots,\,\lambda_{n_r}\}$  and is represented as: 
{\small 
\begin{equation}
\bfG_r(s)=\mathbf{C}_r\left(s\mathbf{I}-
\mathbf{A}_r\right)^{-1}\mathbf{B}_r +\bD_r
 =\sum_{k=1}^{n_r} \frac{\bc_k\ \bfb_k^T}{s-\lambda_k} +\bD_r \label{GrRepr}
 \end{equation}} where $\bfA_r \in \IR^{n_r \times n_r}$ and $\bfB_r\in {\mathbb
R}^{n_r\times
m}$, and  $\bfC_r\in {\mathbb R}^{p\times n_r}.$
Then $\bfG_r$ must satisfy for each $k=1,\,\ldots,\,n_r$,
\begin{subequations}\label{eq:weight_intp_cond}
\begin{align}
 \mathfrak{F}[\bfG](-\lambda_k)\bfb_k &= \mathfrak{F}[\bfG_r](-\lambda_k) \bfb_k
\label{rightInterpCond}\\[.1in]
\bc_k^T\mathfrak{F}[\bfG](-\lambda_k) &=\bc_k^T \mathfrak{F}[\bfG_r](-\lambda_k), \mbox{ and }
\label{leftInterpCond} \\[.1in]
\bc_k^T\mathfrak{F}^{\,\prime}[\bfG](-\lambda_k)\bb_k &=
\bc_k^T\mathfrak{F}^{\,\prime}[\bfG_r](-\lambda_k) \bb_k.    \label{bitanInterpCond} 
 \end{align}
 \setcounter{subeqnSave}{\value{equation}}  
 \addtocounter{parentequation}{-1}                 
 \setcounter{parenteqnSave}{\value{parentequation}}
  \end{subequations}
where $\mathfrak{F}$ is defined in (\ref{Fmap}) and 
$\mathfrak{F}^{\,\prime}[\,\cdot\,](s) = \frac{d\
}{ds}\mathfrak{F}[\,\cdot\,](s)$.
\end{theorem}
(Theorem \ref{thm:IRFInterpCond} provides one additional condition.)
\begin{proof}
Pick an arbitrary vector 
$\bg\in \IC^{p}$ with $\|\bg\|=1$ and an index $k$ with $1\leq k\leq n_r$. 
Suppose that  
\begin{equation*} 
\left\langle \bfG-\bfG_r,\ \frac{\bg\bfb_k^T}{s-\lambda_k}
\right\rangle_{\!\cHtwoW} \!\!= \alpha_0 \neq 0.
\end{equation*}
Define $\theta_0=\arg(\alpha_0)$ and for arbitrary $\varepsilon>0$, define a
perturbation to $\bfG_r$ as
$$
\widetilde{\bfG}_{r}^{(\varepsilon)}(s)= \frac{\bfc_k+\varepsilon\,e^{-\imath
\theta_0}\bg}{s-\lambda_k} \bfb_k^T
+ \sum_{i\neq k} \frac{\bfc_i\bfb_i^T}{s-\lambda_i}.
$$
Then, using (\ref{eq:absH2winner}) and Proposition \ref{G1G2H2}, we obtain
{\small 
\begin{align*}
\|\bfG_r-\widetilde{\bfG}_{r}^{(\varepsilon)}\|_{\cHtwoW}&=\left\|\frac{
-\varepsilon\,e^{-\imath \theta_0}}{s-\lambda_k}\bg\bfb_k^T \right\|_{\cHtwoW}
\\ &\leq \|\bfW\|_{\cHinf}\frac{
\|\bfb_k\|\varepsilon}{\sqrt{2|\real(\lambda_k)|}}.
\end{align*}
}{\small Thus, $\| \bfG_r(s)-\widetilde{\bfG}_{r}^{(\varepsilon)}(s)
\|_{\cHtwoW}
=\mathcal{O}(\varepsilon)$ as $ \varepsilon\rightarrow 0$.
Since $\bfG_r$ solves (\ref{eq:H2optProb}), }
{\small
\begin{align*}
    \| \bfG &- \bfG_r \|_{\cHtwoW}^{2} \leq
    \| \bfG - \widetilde{\bfG}_{r}^{(\varepsilon)} \|_{\cHtwoW}^{2}  \\ 
    &\leq \| (\bfG - \bfG_r) + (\bfG_r- \widetilde{\bfG}_{r}^{(\varepsilon)}) 
\|_{\cHtwoW}^{2}\\
    &\leq\| \bfG - \bfG_r \|_{\cHtwoW}^{2}
    + 2\real \left\langle
\bfG-\bfG_r,\,\bfG_r-\widetilde{\bfG}_{r}^{(\varepsilon)}\right\rangle_{
\!\cHtwoW} \\
& \quad +\| \bfG_r- \widetilde{\bfG}_{r}^{(\varepsilon)} \|_{\cHtwoW}^{2}.
\end{align*}
}{\small 
Thus, 
$$
0\leq 2\ \real \left\langle
\bfG-\bfG_r,\,\bfG_r-\widetilde{\bfG}_{r}^{(\varepsilon)}\right\rangle_{
\!\cHtwoW} +
\| \bfG_r - \widetilde{\bfG}_{r}^{(\varepsilon)} \|_{\cHtwoW}^{2}. $$}\\
This implies that $0\leq  -\varepsilon |\alpha_0| +\mathcal{O}(\varepsilon^2)$, 
which then leads to a contradiction; it must be that $\alpha_0=0$. But then
{\small
\begin{align*}
0&=\left\langle \bfG-\bfG_r,\ \frac{\bg\bfb_k^T}{s-\lambda_k}
\right\rangle_{\!\cHtwoW}  = \left\langle \mathfrak{F}[\bfG-\bfG_r],\
\frac{\bg\bfb_k^T}{s-\lambda_k}
\right\rangle_{\cHtwo} \\
 &= \bg^T\left(\mathfrak{F}[\bfG-\bfG_r](-\lambda_k)\right)\bfb_k, 
\end{align*}}
(using Proposition \ref{G1G2H2}) and since $\bfg$ was chosen
arbitrarily,
we must have {\small 
$$
0=\mathfrak{F}[\bfG-\bfG_r](-\lambda_k)\bfb_k=
\mathfrak{F}[\bfG](-\lambda_k)\bfb_k-\mathfrak{F}[\bfG_r](-\lambda_k)\bfb_k
$$ }which confirms (\ref{rightInterpCond}).  (\ref{leftInterpCond}) is shown
similarly, replacing 
$\frac{\bg\bfb_k^T}{s-\lambda_k}$ in the argument above with
$\frac{\bfc_k\bg^T}{s-\lambda_k}$ for arbitrary $\bg\in \IC^{m}$. 

To show (\ref{bitanInterpCond}), suppose that 
$
\left\langle \bfG-\bfG_r,\
\frac{\bfc_k\bfb_k^T}{(s-\lambda_k)^2}\right\rangle_{\!\cHtwoW} \!\!= \alpha_1
\neq 0.
$
and define $\theta_1=\arg(\alpha_1)$. 
For $\varepsilon>0$  sufficiently small, define
$$
\widetilde{\bfG}_{r}^{(\varepsilon)}(s)=
\frac{\bfc_k\bfb_k^T}{s-(\lambda_k+\varepsilon\,e^{-\imath \theta_1})} 
+ \sum_{i\neq k} \frac{\bfc_i\bfb_i^T}{s-\hat{\lambda}_i} 
$$
As $\varepsilon\rightarrow 0$, we have {\small 
\begin{align*}
\|\bfG_r-\widetilde{\bfG}_{r}^{(\varepsilon)}\|_{\cHtwoW}&=\left\|\frac{
-\varepsilon\,e^{-\imath \vartheta_1}\bfc_k\bfb_k^T}
{(s-\lambda_k)(s-(\lambda_k+\varepsilon\,e^{-\imath
\theta_1}))}\right\|_{\cHtwoW}  \\
&=\mathcal{O}(\varepsilon)
\end{align*} }
Following a similar argument as before, we find that $0\leq  -\varepsilon
|\alpha_1| +\mathcal{O}(\varepsilon^2)$ as 
$\varepsilon\rightarrow 0$,  which leads to a contradiction, forcing
$\alpha_1=0$.  This, in turn, implies from  Proposition \ref{G1G2H2},
{\small
\begin{align*}
0&=\left\langle \bfG-\bfG_r,\, \frac{\bfc_k\bfb_k^T}{(s-\lambda_k)^2}
\right\rangle_{\!\cHtwoW} 
\hspace*{-3ex} = \left\langle \mathfrak{F}[\bfG-\bfG_r],\,
\frac{\bfc_k\bfb_k^T}{(s-\lambda_k)^2} \right\rangle_{\cHtwo} \\
 &\qquad = -\frac{d\
}{ds}\left.\bc_k^T\left(\mathfrak{F}[\bfG-\bfG_r](s)\right)\bfb_k\right|_{
s=-\lambda_k}, 
\end{align*}
}
which gives (\ref{bitanInterpCond}).$\quad \Box$
\end{proof}

We have one additional necessary condition for optimality that arises 
from the presence of the weighting filter.  For $\bfG,\ \bfG_r\in \cHtwoWmed$, 
let $\bfF(t)$ and $\bfF_r(t)$ denote the impulse response functions associated 
respectively with 
$\mathfrak{F}[\bfG](s)$ and $\mathfrak{F}[\bfG_r](s)$.  That is, 
$\mathfrak{F}[\bfG]=\mathcal{L}\left\{\bfF\right\}$ and 
$\mathfrak{F}[\bfG_r]=\mathcal{L}\left\{\bfF_r\right\}$, where 
$\mathcal{L}\left\{\cdot\right\}$ is the Laplace transform. 

\begin{theorem} \label{thm:IRFInterpCond}
Assume the hypotheses and notation of Theorem \ref{thm:weight_intp_cond}. 
Then for all $\bn\in \mathsf{Ker}(\bD_w^T)$, 
\setcounter{equation}{\value{parenteqnSave}}
\begin{subequations} 
\setcounter{equation}{\value{subeqnSave}}
\begin{equation} \label{IRFInterpCond}
\bfF(0)\bn=\bfF_r(0)\bn.
\end{equation}
\end{subequations} 
\end{theorem}
\begin{proof} 
Pick $\bm\in {\mathbb R}^{p}$ and $\bn\in \mathsf{Ker}(\bD_w^T)$, arbitrarily. 
 From (\ref{WPoleRes}),  $\displaystyle \bm\, \bn^T\ \bfW(s) =\sum_{k=1}^{n_w} (\bn^T\be_k)\ \frac{\bm\,  \bff_k^T}{s-\gamma_k} $ is evidently an $\cHtwo$ function.  Hence,  $\bm\, \bn^T\in \cHtwoWmed $.
 Suppose that  
\begin{equation*}  
\left\langle \bfG-\bfG_r,\ \bm\, \bn^T
\right\rangle_{\!\cHtwoW} \!\!= \alpha_0 \neq 0.
\end{equation*}
Define $\theta_0=\arg(\alpha_0)$ and for arbitrary $\varepsilon>0$, define a
perturbation to $\bfG_r$ as
$$
\widetilde{\bfG}_{r}^{(\varepsilon)}(s)= \varepsilon\,e^{-\imath
\theta_0}\,\bm\, \bn^T + \bfG_r(s)
$$
Arguments identical to those in the proof of Theorem \ref{thm:weight_intp_cond} lead to 
$$
0\leq -2\ \real \left\langle
\bfG-\bfG_r,\,\varepsilon\,\bm\, \bn^T \right\rangle_{\!\cHtwoW} +
\| \varepsilon\,\bm\, \bn^T \|_{\cHtwoW}^{2},
$$
implying that $0\leq  -\varepsilon |\alpha_0| +\mathcal{O}(\varepsilon^2)$, 
and leading to a contradiction as before; as a consequence, $\alpha_0=0$. But then
{\small 
\begin{align*}
0&=\left\langle \bfG-\bfG_r,\ \bm\, \bn^T \right\rangle_{\!\cHtwoW}
= \left\langle \mathfrak{F}[\bfG-\bfG_r],\
\bm\, \bn^T \right\rangle_{\cHtwo} \\
 &\quad = \bm^T \left[ \int_{-\infty}^{+\infty} \mathfrak{F}[\bfG-\bfG_r](\imath\omega)\,
\mathrm{d}\omega\right]\bn.
\end{align*}}
Since $\bm$ was chosen arbitrarily,
we must have {\small 
$$
\bzer = \left[ \int_{-\infty}^{+\infty}
\mathfrak{F}[\bfG-\bfG_r](\imath\omega)\,
\mathrm{d}\omega\right]\bn =\left[ \bfF(0)-\bfF_r(0)\right]\bn.
$$ }which confirms (\ref{IRFInterpCond}).  $\quad \Box$
\end{proof}

\section{The Halevi optimality conditions}  \label{sec:halevi}

Following \cite[Appendix A]{Hal92}, the first-order necessary conditions for a
locally optimal reduced model $\bG_r$ can be stated in terms of
solutions to linear matrix equations. Consider the set of matrix equations
defined by $\bG,\bG_r \in \cHtwoWmed $ and $\bW\in\mathcal{H}_{\infty}$ as follows: {\small
\begin{subequations}
  \begin{align}
     \cbfA_{\mathfrak{F}}
    \bX+
 \bX\bA_r^T +  
 \cbfB_{\mathfrak{F}}
\bB_r^T = \bzer,\label{eq:Hal_mat_eq_b}
    \end{align}
    \begin{equation}
    \begin{aligned}
     \bA_r \bP_r& + \bP_r\bA_r^T + \bB_r \begin{bmatrix}\bzer
    & \bC_w \end{bmatrix} \bX     \\  +& \left(\bX^T \begin{bmatrix} \bzer
\\  \bC_w^T
    \end{bmatrix} + \bB_r \bD_w \bD_w^T \right) \bB_r^T =\bzer,
    \label{eq:Hal_mat_eq_c}
    \end{aligned}
    \end{equation}
    \begin{align}
    \bA_r^T \bQ_r + \bQ_r \bA_r + \bC_r^T\bC_r
    &=\bzer,\label{eq:Hal_mat_eq_d} 
    \end{align}
  \begin{equation}
     \begin{aligned}
     \cbfA_{\mathfrak{F}}^T
    \bY +
\bY\bA_r   
= 
\begin{bmatrix} \bC^T \\ ((\bD-\bD_r)\bC_w)^T \end{bmatrix}
\bC_r - 
    \begin{bmatrix}\bzer \\
    \bC_w^T \end{bmatrix}\bB_r^T\bQ_r  .\label{eq:Hal_mat_eq_e}
  \end{aligned}
    \end{equation}  
\end{subequations}}
If $\bG_r$ is locally $\cHtwoWmed $-optimal,
then:{\small 
\begin{subequations}\label{eq:Halevi_opt}
  \begin{equation}
  \begin{aligned}
    \bY^T \bX + \bQ_r \bP_r &= \bzer,\label{eq:Halevi_opt_a}
    \end{aligned}
    \end{equation}
    \begin{equation}
    \begin{aligned}
     \cbfC_{\mathfrak{F}}
     \bX - \bC_r\bP_r -\bD_r
\begin{bmatrix} \bzer & \bC_w
\end{bmatrix} \bX &=\bzer,\label{eq:Halevi_opt_b}
\end{aligned}
\end{equation}
\begin{equation}
\begin{aligned}
       \bY^T& 
       \cbfB_{\mathfrak{F}}
+\bQ_r \left( \bB_r
\bD_w\bD_w^T +
\bX^T     \begin{bmatrix}\bzer \\ \bC_w^T \end{bmatrix}\right)    
=\bzer
    ,\label{eq:Halevi_opt_c} 
    \end{aligned}
    \end{equation}
    \begin{equation}
    \begin{aligned}
    \bC_r \bX^T \begin{bmatrix} \bzer \\
\bC_w^T \end{bmatrix}\bN - \bC\bZ\bC_w^T\bN = ( \bD -\bD_r) \bC_w \bP_w
\bC_w^T\bN  ,\label{eq:Halevi_opt_d}
  \end{aligned} 
  \end{equation}
\end{subequations}}where $\bN=[\bn_1,\dots,\bn_{\ell}]$ is a basis for 
$\mathsf{Ker}(\bD_w^T).$ 

Notice that for $\bW(s)=\bI,$ conditions (\ref{eq:Halevi_opt_a})-(\ref{eq:Halevi_opt_c})
coincide with the Wilson optimality conditions from \cite{Wil70}, while the final condition (\ref{eq:Halevi_opt_d}) is satisfied vacuously since in this case, $\mathsf{Ker}(\bD_w^T)=\{0\}$.


\subsection{Equivalence of the optimality conditions}

The close connection between Sylvester
equations and tangential interpolation in the unweighted case has been established  in \cite{GalVV04}. The model reduction bases that enforce tangential interpolation can be obtained as solutions to
special Sylvester equations. Moreover,  in \cite{GugAB08}, the necessary $\mathcal{H}_2$ optimality conditions in the form of Sylvester equations  from \cite{Wil70} have been
shown to be equivalent to the interpolatory conditions from \cite{MeiL67,GugAB08}.  For the weighted case,  there are two frameworks as well: the interpolatory conditions \eqref{rightInterpCond}-\eqref{IRFInterpCond} we developed here and the linear matrix equations based conditions \eqref{eq:Halevi_opt_a}-\eqref{eq:Halevi_opt_d} of  Halevi \cite{Hal92}. Since these are only  
necessary conditions, their equivalence  is not obvious. 
We formally establish this equivalency.

\begin{theorem}\label{thm:equiv_opt_cond}
Let $\bG,\bG_r \in \cHtwoWmed $ and $\bW\in {{\mathcal H}_{\infty}}.$ Assume that
$\bG_r$ has simple poles at $\{\lambda_1,\dots,\lambda_{n_r}\}.$
Then optimality conditions \eqref{rightInterpCond}-\eqref{IRFInterpCond} and \eqref{eq:Halevi_opt_a}-\eqref{eq:Halevi_opt_d}
are equivalent.

\end{theorem}
\begin{proof}
  Assume $\bG_r$ satisfies \eqref{eq:Halevi_opt_a}-\eqref{eq:Halevi_opt_d} and that $\bA_r = \bR \bLamb
\bR^{-1}$ is an eigenvalue decomposition of $\bA_r.$ Multiplying
\eqref{eq:Hal_mat_eq_b} with $\bR^{-T}$ from  right gives
 \begin{align*}
 \cbfA_{\mathfrak{F}}
  \tilde{\bX} +
\tilde{\bX} \bLamb + 
\cbfB_{\mathfrak{F}}
\tilde{\bB}& = \bzer,
\end{align*}
where $\tilde{\bX}=\bX \bR^{-T}$ and $\tilde{\bB} = \bB_r^T\bR^{-T}.$ 
This implies
{\small
\begin{equation}
\begin{aligned}\label{eq:Halevi_x}
\tilde{\bX}\, \bs_k = 
 \tilde{\bX}_k = 
 (-\lambda_k\bI-\cbfA_{\mathfrak{F}})^{-1}
 \cbfB_{\mathfrak{F}}
\bb_k,
\end{aligned}
\end{equation}
}where $\bs_k$ is the $k^{\rm th}$ unit vector.
Similarly, multiplying \eqref{eq:Hal_mat_eq_c} from  right with
$\bR^{-T}$ yields
\begin{align*}
  \bA_r \tilde{\bP} + \tilde{\bP}\bLamb + &\bB_r  
\begin{bmatrix}\bzer &
\bC_w
\end{bmatrix} \tilde{\bX} \\ &= - \left(\bX^T \begin{bmatrix} \bzer \\
\bC_w^T 
\end{bmatrix} +
\bB_r \bD_w \bD_w^T \right)\tilde{\bB}   ,
\end{align*}
where $\tilde{\bP}=\bP_r\bR^{-T}.$ 
Since for $\bX=\begin{bmatrix} \bX_1 \\ \bX_2 \end{bmatrix}$
we can conclude that $\bX_2 = \bZ_r^T,$ 
where $\bZ_r$ satisfies
\begin{equation} \label{Z_rDef}
  \bA_r   \bZ_r + \bZ_r \bA_w^T + \bB_r (\bC_w\bP_w+\bD_w\bB_w^T) =\bzer. 
  \end{equation}
  It also follows
\begin{align}
 \tilde{\bP}\,\bs_k =
 \tilde{\bP}_k &= (-\lambda_k \bI -\bA_r)^{-1} (\bZ_r\bC_w^T
+\bB_r\bD_w
\bD_w^T )\bb_k \nonumber\\ 
& \quad +(-\lambda_k \bI -\bA_r)^{-1}   \bB_r
\bC_w(-\lambda_k \bI - \bA_w)^{-1} \nonumber\\
 & \qquad \times(\bP_w \bC_w^T + \bB_w
\bD_w^T ) \bb_k.
\end{align}
Right multiplication of \eqref{eq:Halevi_opt_b} with
$\bR^{-T},$ gives 
$$
\cbfC_{\mathfrak{F}}
\tilde{\bX} - \bC_r \tilde{\bP}
-\bD_r \begin{bmatrix} \bzer & \bC_w
\end{bmatrix}\tilde{\bX}= \bzer.$$
Hence, due to Lemma \ref{lem:conn_lcs_F}, each column is
equivalent to \eqref{rightInterpCond}.
Now postmultiply \eqref{eq:Hal_mat_eq_d} with $\bR$ to obtain
\begin{align*}
  \bA_r^T \tilde{\bQ} + \tilde{\bQ}\bLamb  + \bC_r^T \tilde{\bC} = \bzer,
\end{align*}
where $\tilde{\bQ} = \bQ_r \bR$ and $\tilde{\bC} = \bC_r\bR.$ Hence, it
follows \begin{align}\label{eq:Halevi_q}
  \tilde{\bQ}\,\bs_k \tilde{\bQ}_k = (-\lambda_k \bI - \bA_r^T )^{-1} \bC_r^T\bc_k.
\end{align}
Also, postmultiplication of \eqref{eq:Hal_mat_eq_e}
with $\bR$ leads to
{\small
\begin{align*}
\cbfA_{\mathfrak{F}}^T
\tilde{\bY} +
\tilde{\bY} \bLamb 
=  \begin{bmatrix} \bC^T \\ ((\bD-\bD_r)\bC_w)^T
\end{bmatrix}
\tilde{\bC} -
    \begin{bmatrix}\bzer \\
    \bC_w^T \end{bmatrix}\bB_r^T\tilde{\bQ}  
\end{align*}}where $\tilde{\bY}= \bY\bR.$ In particular, we get 
{\small
\begin{align}\label{eq:Halevi_y}
 \tilde{\bY}\,\bs_k =&\tilde{\bY}_k =  
 (-\lambda_k\bI-\cbfA_{\mathfrak{F}})^{-T}
\\  \nonumber
&\times \left( \begin{bmatrix} \bzer \\ \bC_w^T \end{bmatrix}  \bB_r^T
(-\lambda_k \bI-\bA_r^T )^{-1} \bC_r^T +\bD_r^T
-\cbfC_{\mathfrak{F}}^T
\right)\bc_k.
\end{align}
}
We further have 
$
  \tilde{\bY}^T 
  \cbfB_{\mathfrak{F}}
+ \tilde{\bQ} \left( \bB_r \bD_w\bD_w^T
+  \bZ_r \bC_w^T \right) =\bzer$ due to \eqref{eq:Halevi_opt_c}.
Together with \eqref{eq:Halevi_q} and \eqref{eq:Halevi_y}, for each row it
thus holds
\begin{align*}
 \bzer&=
 -\bc_k^T\cbfC_{\mathfrak{F}}
(-\lambda_k\bI-\cbfA_{\mathfrak{F}})^{-1}\cbfB_{\mathfrak{F}}
 \\ &\ \ + \bc_k^T(\bC_r
(-\lambda_k \bI-\bA_r)^{-1}\bB_r+\bD_r) \\ &  \qquad \qquad \times  \bC_w 
(-\lambda_k \bI - \bA_w)^{-1} (\bB_w\bD_w^T+ \bP_w
\bC_w^T)\\
& \ \  +\bc_k^T\bC_r (-\lambda_k \bI -\bA_r)^{-1} \left(
\bB_r\bD_w \bD_w^T + \bZ_r\bC_w^T \right).
\end{align*}
Again, using Lemma \ref{lem:conn_lcs_F}, this
leads to  (\ref{leftInterpCond}).
Finally, pre- and postmultiplication of \eqref{eq:Halevi_opt_a} with
$\bR^T$ and $\bR^{-T}$ yields
\begin{align}\label{eq:Halevi_yx}
 \tilde{\bY}^T\tilde{\bX}+\tilde{\bQ}\tilde{\bP} = \bzer.
\end{align}
Using \eqref{eq:Halevi_x} - \eqref{eq:Halevi_y} for the diagonal of
\eqref{eq:Halevi_yx}, we find
\begin{align*}
\bzer &=-\bc_k^T 
\cbfC_{\mathfrak{F}}(-\lambda_k\bI-\cbfA_{\mathfrak{F}})^{-2}\cbfB_{\mathfrak{F}}
\bb_k \\
&\ \ + \bc_k^T(\bC_r (-\lambda_k \bI -
\bA_r)^{-1} \bB_r+\bD_r) \\ & \quad \qquad \times \bC_w (-\lambda_k \bI -
\bA_w)^{-2} ( \bB_w \bD_w^T +\bP_w \bC_w^T )
\bb_k
\\
&\ \ + \bc_k^T\bC_r (-\lambda_k \bI -\bA_r)^{-2} \left(\bZ_r
\bC_w^T + \bB_r \bD_w \bD_w^T \right) \\
&\ \ +\bc_k^T\bC_r (-\lambda_k \bI -\bA_r)^{-2} \bB_r \\ & \quad \qquad \times
\bC_w (-\lambda_k \bI - \bA_w)^{-1}(\bB_w \bD_w^T+ \bP_w \bC_w^T) \bb_k.
\end{align*}
Then, due to Lemma \ref{lem:conn_lcs_F}, 
 this implies \eqref{bitanInterpCond}.
Finally, due to \eqref{eq:Halevi_opt_d} we note that 
\begin{align*}
\begin{bmatrix} \bC_r &\bD_r \bC_w \end{bmatrix} \begin{bmatrix} \bZ_r
\bC_w^T \\ \bP_w \bC_w^T \end{bmatrix}\bN
= \begin{bmatrix} \bC & \bD \bC_w \end{bmatrix} \begin{bmatrix} \bZ \bC_w^T
\\ \bP_w \bC_w^T \end{bmatrix}\bN.
\end{align*}
From \cite{GugAB08},  $\int_{-\infty}^{\infty} (i \omega \bI
-\bM)^{-1}\, \mathrm{d}\omega = \pi \bI,$ for any stable matrix $\bM,$  and we
 conclude that 
 {\small
\begin{align*}
 & \frac{1}{\pi} \int_{-\infty}^{\infty}  \begin{bmatrix} \bC_r &\bD_r \bC_w
\end{bmatrix} \begin{bmatrix} i \omega \bI -\bA_r & -\bB_r\bC_w \\ \bzer & 
\imath\omega \bI - \bA_w \end{bmatrix} ^{-1}\begin{bmatrix} \bZ_r
\bC_w^T \\ \bP_w \bC_w^T \end{bmatrix}\bN \, \mathrm{d}{\omega} \\ & \ =
\frac{1}{\pi}
\int_{-\infty}^{\infty}  \begin{bmatrix} \bC &\bD \bC_w
\end{bmatrix} \begin{bmatrix} i \omega \bI -\bA & -\bB\bC_w \\ \bzer & 
\imath\omega \bI - \bA_w \end{bmatrix} ^{-1}\begin{bmatrix} \bZ \bC_w^T \\ \bP_w
\bC_w^T \end{bmatrix}\bN \, \mathrm{d}{\omega}
\end{align*}}
Hence, for all $\bn \in \mathsf{Ker}(\bD_w^T),$
$$\left[\int_{-\infty}^{\infty} \mathfrak{F}[\bG_r](\imath\omega)\,
\mathrm{d}\omega\right]\bn = \left[\int_{-\infty}^{\infty}
\mathfrak{F}[\bG](\imath\omega)\,
\mathrm{d}\omega\right]\bn,$$ 
which is equivalent to (\ref{IRFInterpCond}). Reversing the
arguments and using \eqref{eq:resolvent_ident} for the
offdiagonal entries of \eqref{eq:Halevi_opt_a} shows that
\eqref{rightInterpCond}-\eqref{IRFInterpCond} also imply
\eqref{eq:Halevi_opt_a}-\eqref{eq:Halevi_opt_d}. $\quad \Box$
 \end{proof}

\section{Frequency-weighted rational interpolation}  \label{sec:nowi}
We henceforth assume that the feedthrough term of
the original system, $\bG$, is zero: $\bD=\bzer.$  This is without loss of
generality 
since the general case may be recovered
 by reassigning 
$\bD_r\leftarrow \bD_r-\bD$.  
From the previous discussion, we have seen that frequency-weighted ${\mathcal
H}_2$-optimal approximants are mapped to Hermite interpolants via the mapping 
$\mathfrak{F}$ introduced in (\ref{Fmap}).   This presents a practical problem
of how to construct reduced order systems, $\bG_r$, such that
$\mathfrak{F}[\bG_r](s)$ interpolates $\mathfrak{F}[\bG](s)$ at selected points
in $\IC$, say at 
$\left\{\sigma_1,\,\sigma_2,\,\ldots,\,\sigma_{n_r}\right\}$, in selected
tangent
directions $\{\bb_1,\dots,\bb_{n_r}\}$ and $\{\bc_1,\dots,\bc_{n_r}\}.$ 
Using the realization developed in Lemma \ref{lem:conn_lcs_F} and standard
interpolation results, 
we construct reduction subspaces that will force interpolation:
{
\begin{equation}\label{eq:rrs}
\begin{aligned}
\Ran\begin{bmatrix}\mathbb{V}^{(a)} \\
\mathbb{V}^{(b)}\end{bmatrix} =\stackrelbelow{i=1,\ldots,
n_r}{\spn}&\left\{
(\sigma_i\bI-\cbfA_{\mathfrak{F}})^{-1}\cbfB_{\mathfrak{F}}
\bb_i\right\}. 
\end{aligned}
\end{equation}
}
and
{
\begin{align}\label{eq:lrs}
\Ran\begin{bmatrix}\mathbb{W}^{(a)} \\
\mathbb{W}^{(b)}\end{bmatrix}=\stackrelbelow{i=1,\ldots, n_r}{\spn}&\left\{
(\sigma_i\bI-\cbfA_{\mathfrak{F}}^T)^{-1}\cbfC_{\mathfrak{F}}^T
\bc_i \right\}.
\end{align}
}
Define $\bV_r,\,\bW_r\in \IC^{n\times n_r}$ so
that 
$\bW_r^T\bV_r=\bI$ and 
\begin{equation} \label{RanVrRanWr}
\begin{array}{c}
\Ran(\mathbf{V}_r)\supset \Ran\left\{\mathbb{V}^{(a)}  \right\}\\
\Ran(\mathbf{W}_r)\supset \Ran\left\{\mathbb{W}^{(a)}  \right\}.
\end{array}
\end{equation}
The reduced feedthrough term is computed from \eqref{eq:Halevi_opt_d}: 
{\small 
\begin{equation}\label{eq:red_feedthr}
  \bD_r = \bC\left(\bZ  -\bV_r \bZ_r\right) \bC_w^T
\bN (\bN^T\bC_w \bP_w \bC_w^T\bN )^{-1} \bN^T,
\end{equation}}
where $\bN$ is a basis for $\mathsf{Ker}(\bD_w^T).$
\begin{theorem}\label{thm:appr_intp}
Let $\bA_r=\bW_r^T \bA \bV_r,\ \bB_r=\bW_r^T\bB,\
\bC=\bC_r\bV_r,$ with $\bV_r$ and $\bW_r$  constructed as in
\eqref{eq:rrs}, \eqref{eq:lrs}, and \eqref{RanVrRanWr}. 
Suppose $\bD_r$ is determined by \eqref{eq:red_feedthr}. 
Then pick any interpolation point 
$\sigma \in \left\{\sigma_1,\,\sigma_2,\,\ldots,\,\sigma_{n_r}\right\}$,
with associated tangent directions: $\bb$ and $\bc$. 
Provided $\sigma \not \in \{\Lambda(\bA),\Lambda(\bA_r)\}$,
we have 
{\small
\begin{subequations}{
\begin{align*}
  \mathfrak{F}[\bG]&(\sigma)\bb- \mathfrak{F}[\bG_r](\sigma)\bb = \\ 
\bH_1&(\sigma)\left(\bZ  -\bV_r \bZ_r\right) \bC_w^T\bb-
\bC\left(\bZ  -\bV_r \bZ_r\right) \bH_2(\sigma)\bb \displaybreak[0]\\[1ex]
\bc^T\mathfrak{F}&[\bG](\sigma)- \bc^T\mathfrak{F}[\bG_r](\sigma)=  \\
\bc^T&\bH_1(\sigma)\left(\bZ  -\bV_r \bZ_r\right) \bC_w^T-
\bc^T\bC\left(\bZ  -\bV_r \bZ_r\right) \bH_2(\sigma),\displaybreak[0]\\[1ex]
  \bc^T\mathfrak{F}'[&\bG](\sigma)\bb -
\bc^T\mathfrak{F}'[\bG_r](\sigma)\bb= \\
\bc^T\bH_1'&(\sigma)\left(\bZ  -\bV_r \bZ_r\right) \bC_w^T \bb-
\bc^T\bC\left(\bZ  -\bV_r \bZ_r\right) \bH_2'(\sigma)\bb,
\end{align*}
and $\qquad\qquad\qquad\bF(0)\bn =\bF_r(0)\bn$,
}
\end{subequations}
}

\vspace{1em}
\noindent
where $\bfF(t)$ and $\bfF_r(t)$ are the impulse responses of 
$\mathfrak{F}[\bfG]$ and $\mathfrak{F}[\bfG_r]$, respectively, 
$\bn \in \mathsf{Ker}(\bD_w^T)$ is arbitrary,
{\begin{align*}
\label{eq:Hr}
\bH_1(s)&=\bC_r(s\bI-\bA_r)^{-1}\bW_r^T, \mbox{ and} \\[1ex]
\bH_2(s) &= \bC_w^T\bN (\bN^T\bC_w \bP_w \bC_w^T\bN )^{-1} \bN^T \times \nonumber \\
 & \qquad \bC_w(s \bI-\bA_w)^{-1}(\bP_w \bC_w^T+ \bB_w\bD_w^T). \nonumber
\end{align*}}
\end{theorem}

\begin{proof}:
We follow a pattern of proof given in \cite{AntBG10}. Define 
$\mathbb{V} =
\begin{bmatrix} \bV_r & \bzer \\ \bzer & \bI \end{bmatrix}$, $\mathbb{W} =
\begin{bmatrix} \bW_r & \bzer \\ \bzer & \bI \end{bmatrix}$, and 
$
\cbfA_\mathfrak{Fr} =
\begin{bmatrix} \bA_r & \bB_r\bC_w \\ \bzer &
\bA_w \end{bmatrix}$.
Define two (skew)
projectors via
\begin{align*}
 {\mathcal P}_r(s) &= \mathbb{V} (s \bI -\cbfA_\mathfrak{Fr})^{-1} \mathbb{W}^T (s \bI - \cbfA_\mathfrak{F}) \\[1ex]
{\mathcal Q}_r(s)&=  (s \bI - \cbfA_\mathfrak{F})  {\mathcal P}_r(s) (s \bI - \cbfA_\mathfrak{F}) ^{-1} \\
&= (s \bI - \cbfA_\mathfrak{F}) 
\mathbb{V} (s \bI -\cbfA_\mathfrak{Fr})^{-1}\mathbb{W}^T.
\end{align*}
 For all $s$ in a neighborhood of
$\sigma,$
we have ${\mathcal V}=\mathsf{Ran}({\mathcal P}_r(s))=\mathsf{Ker}(\bI-{\mathcal P}_r(s))$
and ${\mathcal W}^{\perp}=\mathsf{Ker}({\mathcal Q}_r(s))=\mathsf{Ran}(\bI-{\mathcal
Q}_r(s)).$ 
Now observe that
{\small 
\begin{align*}
 &\mathfrak{F}[\bG_r](s)= \\ & \begin{bmatrix} \bC_r & \bzer \end{bmatrix}
\begin{bmatrix} s\bI-\bA_r & -\bB_r \bC_w \\ \bzer & s\bI-\bA_w
\end{bmatrix}^{-1} \begin{bmatrix} \bW_r^T\bZ\bC_w^T + \bB_r \bD_w \bD_w^T \\
\bP_w\bC_w ^T + \bB_w\bD_w^T \end{bmatrix} \\
 & \quad  - \begin{bmatrix} \bC_r &\bzer \end{bmatrix}
\begin{bmatrix} s\bI-\bA_r & -\bB_r \bC_w \\ \bzer & s\bI-\bA_w
\end{bmatrix}^{-1} \begin{bmatrix} (\bW_r^T\bZ-\bZ_r)\bC_w^T 
\\
 \bzer \end{bmatrix} \\
& \qquad +\bD_r \bC_w(s\bI-\bA_w)^{-1} (\bP_w\bC_w ^T+ \bB_w \bD_w^T) .
\end{align*}}
Hence, we can write {
\begin{equation}\label{eq:FminusFr}
\begin{aligned}
 &\mathfrak{F}[\bG](s) -
\mathfrak{F}[\bG_r](s) \\
& = \bH_1(s)\left(\bZ  -\bV_r \bZ_r\right) \bC_w^T-
\bC\left(\bZ  -\bV_r \bZ_r\right) \bH_2(s)\\
&\quad+\cbfC_{\mathfrak{F}} (s\bI-\cbfA_{\mathfrak{F}})^{-1}(\bI-{\mathcal
Q}_r(s)) 
(s\bI-\cbfA_{\mathfrak{F}})\\
&\qquad \times(\bI-{\mathcal
P}_r(s))(s\bI-\cbfA_{\mathfrak{F}})^{-1}\cbfB_{\mathfrak{F}}
\end{aligned}
\end{equation}
}
Evaluating this expression at $s=\sigma$ and postmultiplying by $\bb$ yields
the first assertion; premultiplying by $\bc^T$ yields the second. We find that 
{\small
\begin{align*}
((\sigma+\varepsilon)\bI-\cbfA_{\mathfrak{F}})^{-1} = (\sigma
\bI-\cbfA_{\mathfrak{F}})^{-1} - \varepsilon
(\sigma\bI-\cbfA_{\mathfrak{F}})^{-2}
+ {\mathcal O}(\varepsilon^2).
\end{align*}}
Evaluating \eqref{eq:FminusFr} at $s=\sigma+\varepsilon,$ premultiplying by
$\bc^T,$ and postmultiplying by $\bb$ together with $\varepsilon\to 0$ yields
the third statement. The last statement results from the proof
of Theorem \ref{thm:equiv_opt_cond} and the fact that $\bN$ is a basis of
$\mathsf{Ker}(\bD_w^T).$ Note also that we have $\bD_r\bD_w=\bzer.$
$\quad \Box$ \end{proof}

Conditions for exact interpolation are now evident:
\begin{corollary} \label{cor:exactness}
Let $\bfG_r$ denote the reduced order model of Theorem \ref{thm:appr_intp}. 
If $\bfG_r$  is stable and $\mathsf{Ran}(\bZ)\subset \mathsf{Ran}(\bV_r)$ then  $\mathfrak{F}[\bG_r]$
is an exact bitangential Hermite interpolant to  $\mathfrak{F}[\bG]$ at each 
interpolation point, $\left\{\sigma_1,\,\sigma_2,\,\ldots,\,\sigma_{n_r}\right\}$ 
in corresponding tangent directions,  
$\{\bb_1,\dots,\bb_{n_r}\}$ and $\{\bc_1,\dots,\bc_{n_r}\}.$ 
\end{corollary}
\begin{proof}:
Note first that under the hypotheses, $\bV_r\bW_r^T\bZ=\bZ$, 
Now, premultiply (\ref{eq:weight_sylv}) by $\bW_r^T$ and subtract (\ref{Z_rDef}) to obtain
$$
  \bA_r  \bW_r^T\left(\bZ-\bV_r\bZ_r\right) + \bW_r^T\left(\bZ-\bV_r\bZ_r\right) \bA_w^T  =\bzer. 
$$
Since $\bA_r$ and $\bA_w$ are both stable, 
$$
\bW_r^T\left(\bZ-\bV_r\bZ_r\right) = \bW_r^T\bZ-\bZ_r=\bzer
$$
and so, $\bZ=\bV_r\bZ_r$.
\end{proof}

The deviation from exact interpolation 
is quantified in Theorem  \ref{thm:appr_intp} and depends on the
deviation of $\bV_r\bZ_r$ from $\bZ.$ 
For shaping filters of modest order with $n_w \ll n,$ exact interpolation can be
induced since one may include $\Ran(\bZ)$ 
in the projection space,  $\Ran(\bV_r).$  

More generally, $\bV_r\bZ_r$ may be viewed as a Petrov-Galerkin approximation to the solution $\bZ$ of the Sylvester equation \eqref{eq:weight_sylv} in the following sense: $\bZ_r$ that solves (\ref{Z_rDef}) is a solution to the problem of finding 
$\cbfZ\in \mathbb{R}^{n_r\times n_w}$ such that with respect to the usual (Euclidean) inner product in $\mathbb{R}^n$,
\begin{equation*}
\begin{split}
\Ran&\left( \bA  \left(\bV_r \cbfZ\right) + \left(\bV_r \cbfZ\right) \bA_w^T + \bB (\bC_w\bP_w+\bD_w\bB_w^T)\right)\\
& \qquad  \perp \Ran\left( \bW_r \right).
\end{split}
\end{equation*}
Since  $m,m_w \ll n$, the singular values of the original solution, $\bZ$, to
\eqref{eq:weight_sylv} will typically decay rapidly
\cite{Gra04,Pen00,Sim13,SorZ02}; there will be good
 low rank approximations to $\bZ$ and among them will be approximations of the
form $\bV_r \cbfZ$.  In our approach, the subspace $\bV_r$ is closely related to a $\cHtwo$ optimal approximation. And in the unweighted case,
projection subspaces associated with $\mathcal{H}_2$-optimal reduced models are
known to yield very accurate approximations 
This has been underlined in \cite{DruKS11,FlaG13} by the fact that the
approximations are equivalent to those obtained from the alternating directions
implicit (ADI) iteration. Moreover,  \cite{BenB14} showed that for
symmetric state space systems, low rank approximations from an
$\mathcal{H}_2$-optimal reduced model in fact locally minimize the energy norm
naturally induced by the corresponding Lyapunov operator. Overall, this leads
to the expectation that as $n_r$ increases, $\bV_r\bZ_r \approx \bZ$. If
furthermore, the interpolation 
 points that determine a reduced model coincide with the reflected poles of the
model,  then Theorem  \ref{thm:appr_intp}
 asserts that the optimality conditions
  \eqref{rightInterpCond}-\eqref{IRFInterpCond} will very nearly be satisfied; 
  the reduced model draws closer to $\cHtwoWmed $-optimality as $n_r$ increases.   
  
 The practical difficulty in constructing such \emph{near optimal} reduced models is that one doesn't know \emph{a priori} how 
 to choose interpolation data determining a reduced model so as to coincide with the reflected poles of the model.  
 The parallel circumstance for (unweighted) optimal $\cHtwo$ model reduction has been largely resolved with 
 an iterative correction process \cite{GugAB08};  we propose an analogous approach here:
\begin{table}[h]  
 \fbox{
  \begin{minipage}[t]{3in}
  \!\textbf{Algorithm} \textsc{{nowi}}:\\
  \centerline{\textbf{Nearly Optimal Weighted Interpolation}}\\
   \vspace*{-.3cm}
  \begin{algorithmic}[1]
    \REQUIRE Interpolation points: $\{\sigma_1,\dots,\sigma_{n_r}\};$\\
     Tangent directions: \ $\tilde{\bB} = \left[\bb_1,\dots,\bb_{n_r}\right]$ \\
     \qquad\qquad\qquad\quad and $\tilde{\bC}
=\left[\bc_1,\dots,\bc_{n_r}\right].$
    \ENSURE $\bA_r$, $\bB_r$, $\bC_r$, $\bD_r$   \\
    \WHILE{relative change in $\{\sigma_i\}> \mathrm{tol}$}
    \STATE {\small \hspace{-2ex}Compute $\bV_r$ and $\bW_r$ from \eqref{eq:rrs},  \eqref{eq:lrs}, and \eqref{RanVrRanWr}.
}
    \STATE {\footnotesize \hspace{-3ex}  Update ROM: $\bA_r = \bW_r^T \bA \bV_r,\bB_r = \bW_r^T\bB,$\\
     $\bC_r=\bC\bV_r$, and $\bD_r$ as in \eqref{eq:red_feedthr}}.
    \STATE {\footnotesize \hspace{-3ex} $\sigma_i=-\lambda_i
\left(\mathbf{\Lambda}\right),\bA_r=\bR \mathbf{\Lambda} \bR^{-1},\tilde{\bB}=\bB_r^T\bR^{-T}$, \\
and  $\tilde{\bC}=\bC_r \bR.$ }
    \ENDWHILE\\    
  \end{algorithmic}
\end{minipage}
} 
\end{table}
%

Note that \textsc{nowi} is not simply a MIMO extension of \textsc{wirka} in \cite{Ani13}, which was developed specifically for SISO settings. \textsc{wirka} is heuristic in nature and does not originate from necessary optimality conditions.  On the other hand,  \textsc{nowi} directly attempts to satisfy conditions for optimality and will provide progressively better approximations to them as $n_r$ increases. 
Even in SISO settings, the difference between \textsc{nowi} and \textsc{wirka} is easily seen by noting that the model reduction bases $\bfV_r$ and $\bfW_r$ are completely different.  While \textsc{nowi} uses a state-space realization of $\mathfrak{F}[\bG](s)$  (as the interpolation conditions require) in order to construct $\bfV_r$ and $\bfW_r$, \textsc{wirka} instead uses regular rational Krylov subspaces corresponding to $\bG(s)$ -- generally, not even approximately satisfying the necessary optimality conditions. Moreover, in \textsc{wirka}, $\bfW_r$ is kept constant after initialization unlike in \textsc{nowi} where  both $\bfW_r$  and $\bfV_r$ are updated iteratively.

\paragraph*{Computational complexity:}
Many issues enter in determining the computational resources necessary to produce an effective reduced order model. 
Estimates of computational complexity serve as a useful proxy for this expense, which may be then further refined according to problem-specific structure and implementation. 
Notice first that our \textsc{\textsc{nowi}} Algorithm 
is an iterative process, requiring in each cycle the construction of left- and right- reduction subspaces.
This requires first the solution of two linear matrix equations,
\eqref{eq:weight_lyap} and \eqref{eq:weight_sylv} of orders $n_w\times n_w$ and $n\times n_w$, respectively. 
If $n_w\ll n$, this may be done directly with cost dominated by 
$n_w$ linear solves of dimension $n.$  For larger $n_w$, 
the numerical rank of $\bP_w$ and $\bZ$ is often relatively small allowing for very
accurate approximations by low rank methods such as 
\cite{penzl200clr,gugercin2003amodified,sabino2007solution,benner2008numerical,jaimoukha1994krylov,simoncini2008new}. 
Bases for the left- and right- reduction subspaces then may be computed exploiting 
the block triangular structure of the $\mathfrak{F}$-realization; 
this leads to $2n_r$ linear solves of dimension $n$ and $n_r$ 
linear solves of dimension $n_w$.  Sparsity in $\bA$ and $\bA_w$ may 
be exploited with either direct or iterative linear solvers.  
Multiple right-hand sides and small changes among shifts offer further 
opportunities for efficiency from subspace and preconditioner recycling. 

When compared to standard
approaches for frequency-weighted balanced truncation (\textsc{\textsc{fwbt}}), we find that as
long as the number of iterations of \textsc{\textsc{nowi}} 
remains modest (which appears typical),  the overhead associated with 
solving two large Lyapunov equations of dimension
$n$, which is necessary for \textsc{\textsc{fwbt}}, has been eliminated.   
This creates a particularly dramatic advantage for \textsc{\textsc{nowi}} in the case of a shaping filter where $n_w\ll n.$  
The computational advantages of  
\textsc{\textsc{nowi}} are also significant when 
 compared to Halevi's approach to weighted-$\mathcal{H}_2$ model reduction
\cite{Hal92}, which requires solving large-scale Riccati and Lyapunov equations
of order $(n+n_w) \times (n+n_w)$ at every step of the iteration.

\section{Numerical examples}
\label{sec:num}
\vspace*{-.3cm}
We study the performance of our \textsc{\textsc{nowi}} Algorithm 
for three different
examples resulting from controller reduction. We  compare the proposed method
with frequency weighted balanced truncation (\textsc{fwbt}) of \cite{Enn84}, and
also with \textsc{wirka} of \cite{Ani13} for the SISO example. 

\newlength\fheight
\newlength\fwidth
\setlength\fheight{0.57\linewidth}
\setlength\fwidth{0.76\linewidth}

\paragraph*{Los Angeles University Hospital}
The plant is a
linearized model for the Los Angeles University Hospital with order
$n=48.$   An LQG-based controller of the same order as the
original system is to be reduced, leading to a weighting $W(s)$ of order
$n_w=96$, see \cite{Ani13}. For a given $n_r$, we
use the mirror images of  the $\nu=2$ most dominant  poles of  $W(s)$
and the mirror images $n_r - \nu$ most dominant   poles of  $G(s)$ 
 as  the initial interpolation  points for  \textsc{wirka}, 
 as suggested in \cite{Ani13}. 
We use the same initialization for the \textsc{\textsc{nowi}} Algorithm. 
\begin{figure}[tb]
\begin{center}
 \includegraphics[scale=1.0]{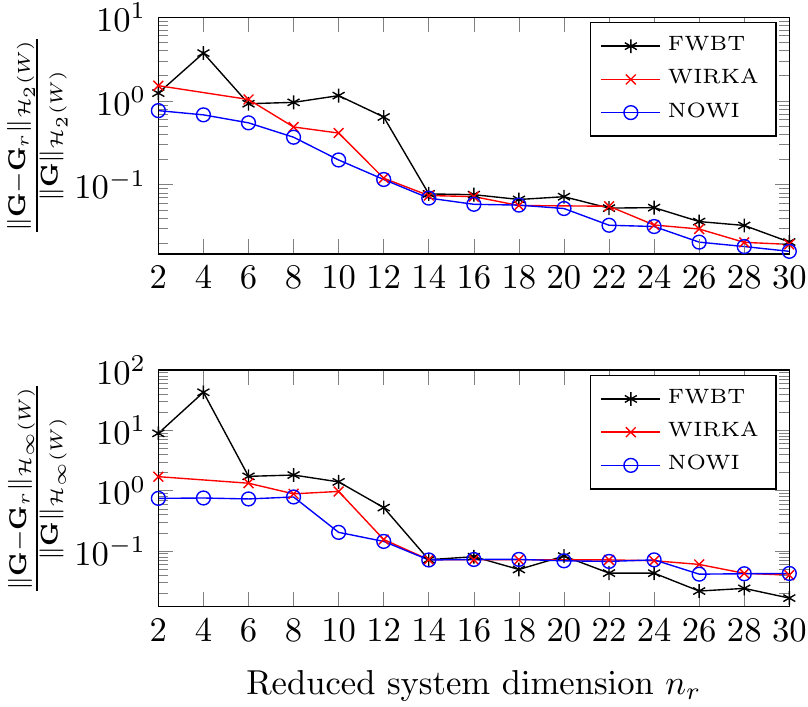}
  \caption{LA university hospital, $n=48,n_w=96.$}
  \label{fig:la_uh_a}
  \end{center}
\end{figure}
Figure \ref{fig:la_uh_a} shows the relative $\cHtwoWmed$- and $\mathcal{H}_{\infty}(W)$-errors obtained from
\textsc{nowi}, \textsc{fwbt}, and \textsc{wirka}
for  reduced system orders $n_r=2,\dots,30$. For the $\cHtwoWmed$-case, \textsc{\textsc{nowi}}  outperforms \textsc{fwbt} and \textsc{wirka} for all $n_r$ values
except for $n_r=18$, for which \textsc{wirka} is slightly better. The
superiority of \textsc{nowi} is especially evident for smaller  $n_r$ values. 
We find similar results for the
$\mathcal{H}_{\infty}(W)$-error as well; \textsc{fwbt} yields the smallest $\mathcal{H}_{\infty}(W)$-errors  for larger $n_r$, as expected. 
The fact that \textsc{nowi} displays better ${\mathcal H}_{\infty}(W)$ performance than \textsc{fwbt} even for  a subset of reduction orders 
suggests the effectiveness of the approach.
\textsc{nowi} produces reduced models that satisfy the $\cHtwoWmed$-optimality interpolation conditions (\ref{eq:weight_intp_cond}) only approximately (see Theorem \ref{thm:appr_intp}).
Figure \ref{fig:la_uh_b} shows how the relative interpolation error (deviation from \eqref{eq:weight_intp_cond}) in final reduced models produced by \textsc{\textsc{nowi}}  
evolves with increasing $n_r$.  As the figure shows, the
relative error in the optimality conditions
 decreases as $n_r$ increases. This
confirms the expectations described in the discussion following Corollary \ref{cor:exactness}.
\begin{figure}[tb]
\begin{center}
  \includegraphics[scale=1.0]{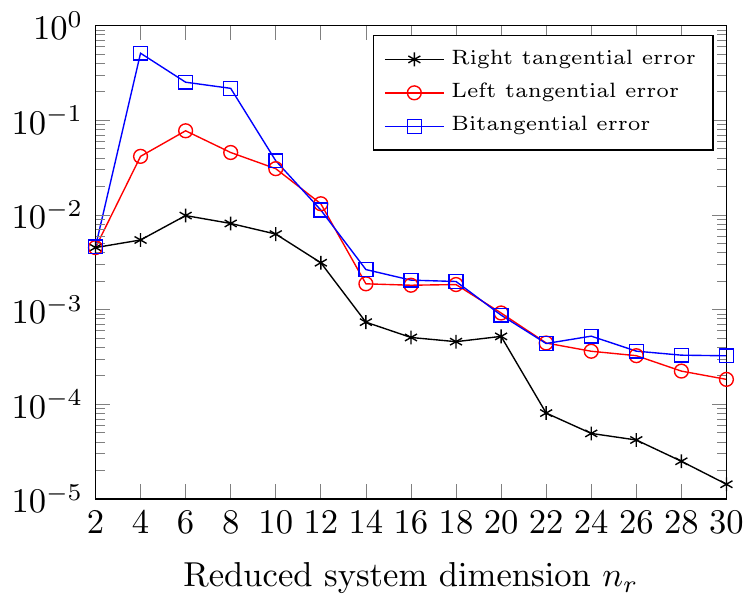}
  \caption{LA university hospital, $n=48,n_w=96.$}
  \label{fig:la_uh_b}
  \end{center}
\end{figure}
Figure \ref{fig:building_b} shows how the relative interpolation error in the 
the optimality conditions \eqref{eq:weight_intp_cond} evolve (for  fixed reduction order, $n_r$) step to step in the \textsc{\textsc{nowi}} Algorithm. 
Results for two cases are displayed: $n_r=16$
and $n_r=30$. In both cases, we observe that \textsc{\textsc{nowi}} 
rapidly reduces  interpolation error during the iteration. For example, for  $n_r=16$,  relative interpolation errors are in the order of $1$ initially; however as the algorithm progresses, relative errors decline to levels of $10^{-3}$,  leading to near-optimal interpolation.
\begin{figure}[tb]
\begin{center}
  \includegraphics[scale=1.0]{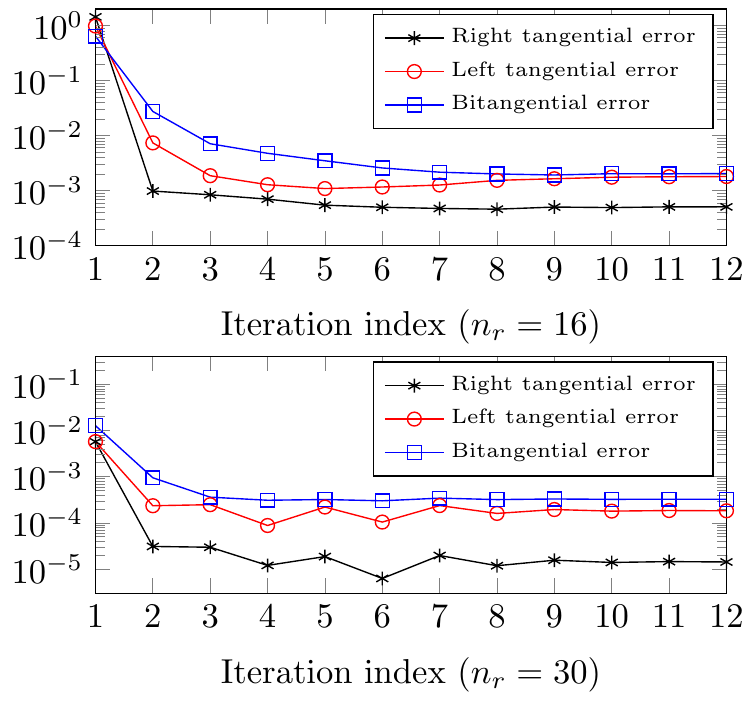}
  \caption{LA university hospital, $n=48,n_w=96.$}
  \label{fig:building_b}
  \end{center}
\end{figure}

\paragraph*{CD player}
The plant is a model for a CD player and belongs to the
\textsc{slicot} benchmark
collection. We consider the original MIMO version with $n=120$ and $m=p=2.$ As
in the previous example, we design an LQG-based controller having the same order as the plant,
leading to a weight $\bW(s)$  with $n_w=240.$
Since \textsc{wirka} has been proposed only for SISO systems and a MIMO extension is not immediate, we
 show comparisons only between \textsc{fwbt} and \textsc{nowi}, using a random initialization. Figure
\ref{fig:cd_a} again compares the quality of reduction in terms of the $\cHtwoWmed$-error and ${\mathcal H}_{\infty}(W)$-error. Both methods perform
equally well with slight advantages for \textsc{\textsc{nowi}} 
 in the case of the $\cHtwoWmed$-error and for \textsc{fwbt} in the case of the ${\mathcal
H}_{\infty}$-error.
\begin{figure}[tb]
\begin{center}
  \includegraphics[scale=1.0]{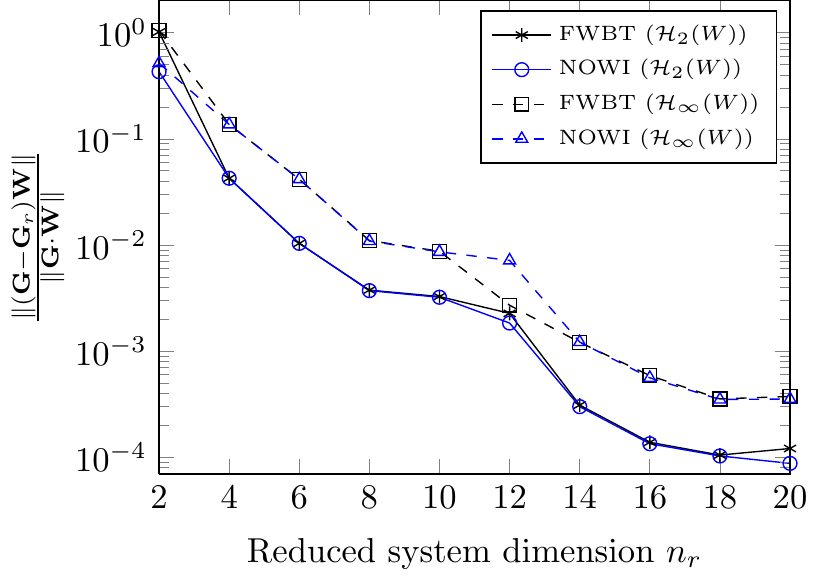}
  \caption{CD player, $n=120,n_w=240.$}
  \label{fig:cd_a}
  \end{center}
\end{figure}
Similar to the previous example,
Figure \ref{fig:cd_interror_vs_r} shows how the relative error in the optimal interpolation conditions
\eqref{eq:weight_intp_cond} vary as $n_r$ varies. 
Once again, the relative residual of the optimality
conditions decreases as $n_r$ increases, yielding near-optimal interpolation.
\begin{figure}[tb]
\begin{center}
  \includegraphics[scale=1.0]{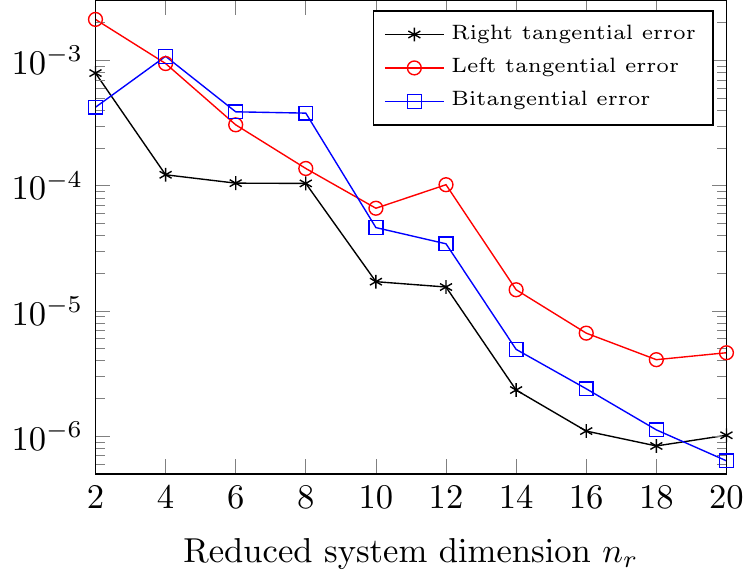}
  \caption{CD player, $n=120,n_w=240.$}
  \label{fig:cd_interror_vs_r}
  \end{center}
\end{figure}

\paragraph*{ISS 1R Module}
The final example is the component 1r of the International Space
Station from
the \textsc{slicot} benchmark collection. The plant is a MIMO system 
with $n=270$, and $m=p=3.$ The controller to be reduced is 
  an LQG-based controller as before. We compare \textsc{nowi} and \textsc{fwbt} 
  for $n_r= 2,4,\ldots,40$. 
 For  $n_r \le 30$, we use logarithmically spaced interpolation points for
initializing \textsc{nowi}. For larger values of $n_r$, we aggregate the optimal points  from
smaller reduced models. 
The  relative $\cHtwoWmed$ errors are shown in Figure
\ref{fig:iss_a}.
The full model is  hard to reduce with slowly decaying Hankel singular values. This is apparent from 
  Figure
\ref{fig:iss_a} where \textsc{fwbt} hardly reduces the error  for smaller $n_r$ values. 
The proposed method clearly  outperforms \textsc{fwbt} for every reduction order.
\begin{figure}[tb]
\begin{center}
  \includegraphics[scale=1.0]{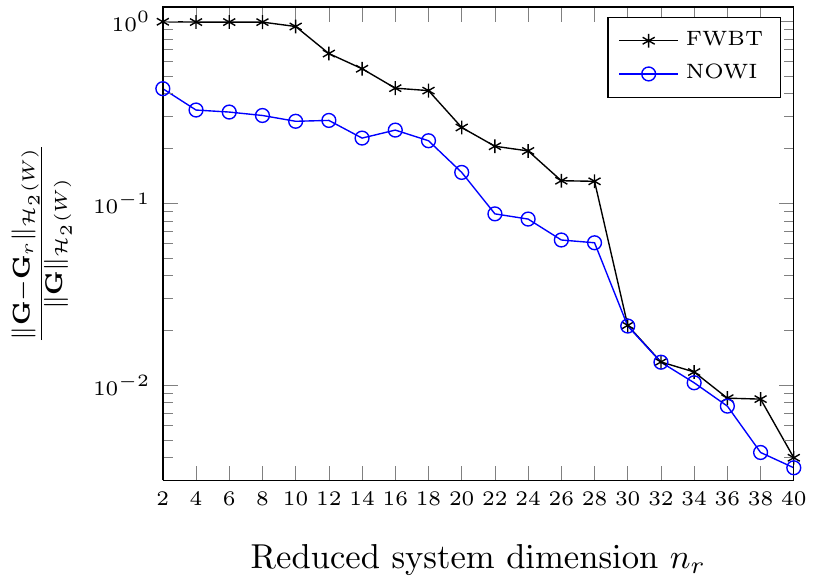}
  \caption{ISS, $n=270,n_w=540.$}
  \label{fig:iss_a}
  \end{center}
\end{figure}

\paragraph*{1-D Beam Model} The full-order model represents the dynamics of a
1-D beam with order $n=3000$ with two inputs  (point forces applied to the
first two states) and one output (the displacement in the middle). The sigma
plot, i.e. $\| \bG(\imath \omega) \|_2$ vs $\omega \in \IR$ is given in Figure
\ref{fig:beam_a}. For the weighting function $\bW(s)$, first we construct an
order $n_w = 60$, two-inputs/two-outputs band-pass filter with $[10^{-3},0.7]$
rad/sec frequency band of interest to focus the emphasis on the first three
peaks in the sigma plot. Using both \textsc{nowi} and \textsc{fwbt}, we reduce
the order to $n_r=16$. \textsc{nowi} was initiated by a random selection of
interpolation points and tangent directions as before. As the Figure
\ref{fig:beam_b} depicts, \textsc{nowi} significantly outperforms
\textsc{fwbt}, successfully achieving high accuracy within the frequency interval of
interest. We repeat the process a using band-pass filter with $[3 \times
10^{-2},0.7]$ rad/sec frequency band of interest. As Figures \ref{fig:beam_c}
and \ref{fig:beam_d} depict, \textsc{nowi} outperforms
\textsc{fwbt} in this case as well. In order to achieve this accuracy, \textsc{nowi} took only $3.34$ seconds to run, while \textsc{fwbt} already took more than $277$
seconds just to solve for the weighted Gramians.
\begin{figure}[tb]
\begin{center}
   \includegraphics[scale=1.0]{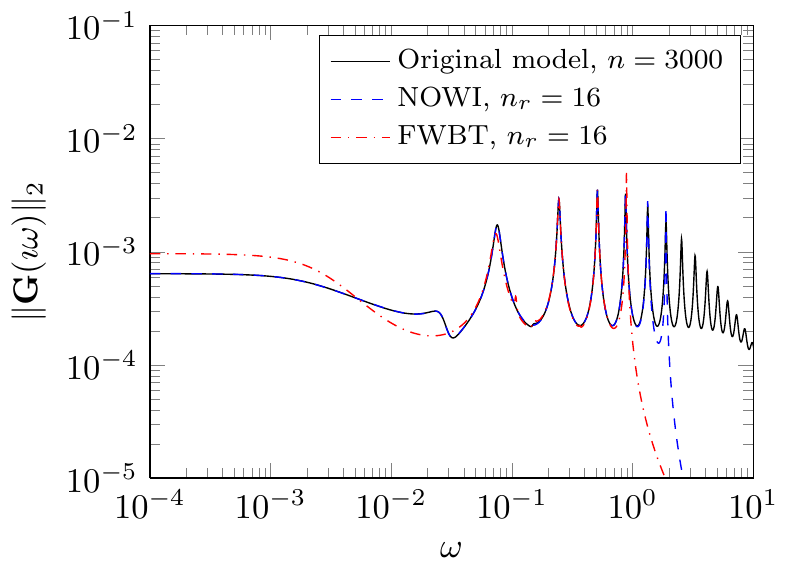}
  \caption{Beam, $n=3000,n_w=60.$}
  \label{fig:beam_a}
  \end{center}
\end{figure}
\begin{figure}[tb]
\begin{center}
   \includegraphics[scale=1.0]{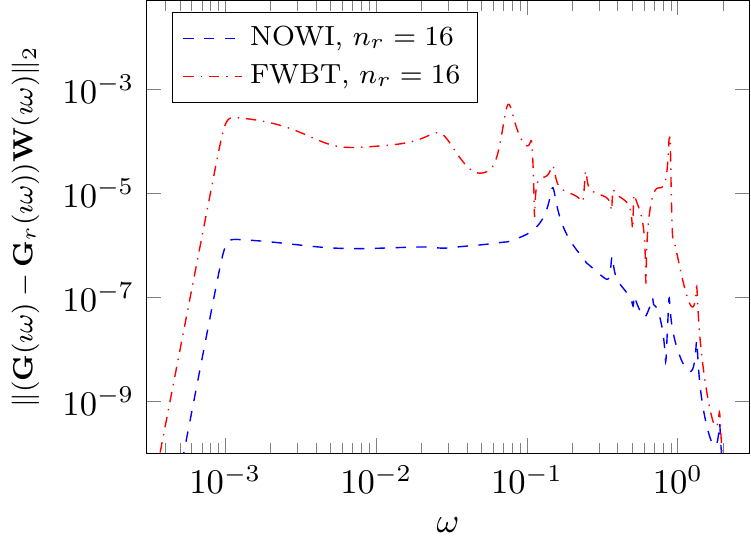}
  \caption{Beam, $n=3000,n_w=60.$}
  \label{fig:beam_b}
  \end{center}
\end{figure}
\begin{figure}[tb]
\begin{center}
   \includegraphics[scale=1.0]{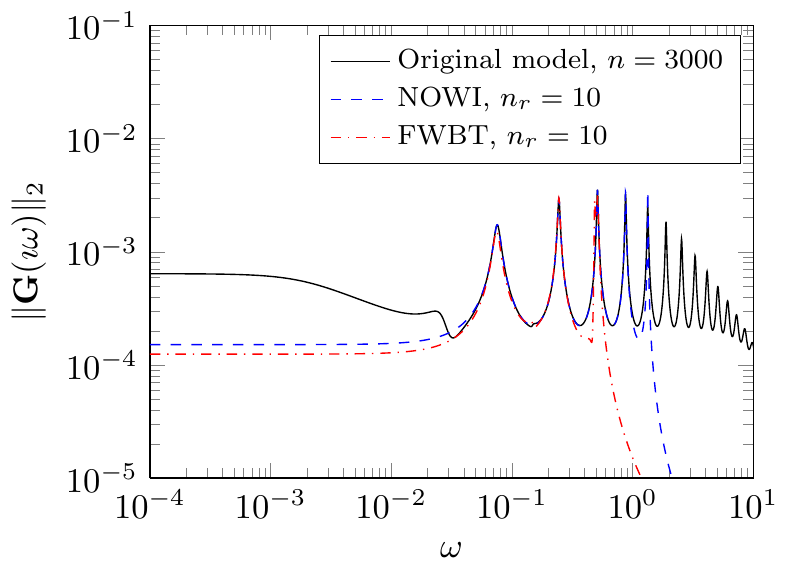}
  \caption{Beam, $n=3000,n_w=60.$}
  \label{fig:beam_c}
  \end{center}
\end{figure}
\begin{figure}[tb]
\begin{center}
   \includegraphics[scale=1.0]{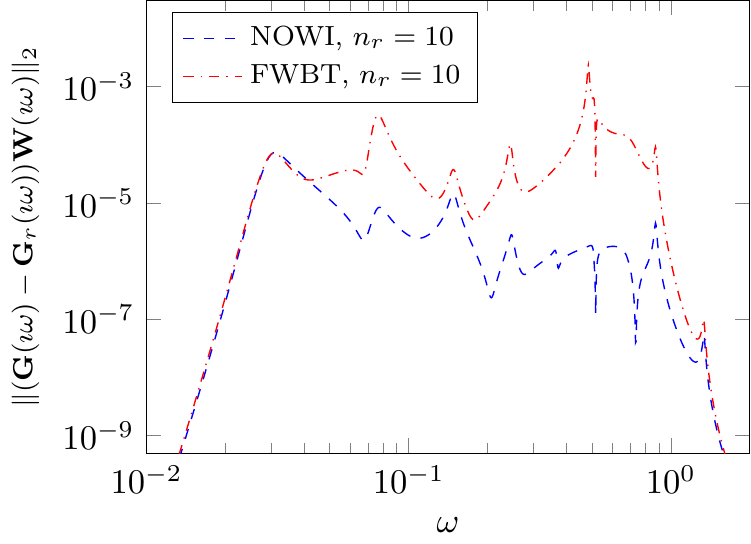}
  \caption{Beam, $n=3000,n_w=60.$}
  \label{fig:beam_d}
  \end{center}
\end{figure}
\section{Conclusions}
\label{sec:conc}
We have extended an interpolatory framework for weighted-$\cHtwo$ model
reduction to include MIMO dynamical systems with feed-forward terms.   The main
tool was a new representation of the weighted-$\cHtwo$ inner product in MIMO
settings (the $\mathfrak{F}$-transformation defined in (\ref{Fmap})) which led to  
associated first-order
necessary conditions that must be satisfied by an optimal 
weighted-$\cHtwo$ reduced-order model.  
These conditions in turn were found to be equivalent to necessary conditions
established earlier by Halevi. An examination of realizations for systems
defined by $\mathfrak{F}[\cdot]$ then led to an algorithm that remains tractable
for large state-space dimension. There are a variety of refinements of the ideas
presented here that can exploit the flexibility afforded by the interpolatory
model reduction framework. One direction that has been fruitful in the
unweighted case is trust-region based descent approaches, as described in
\cite{BeaG09} and extended to frequency-weighted settings in \cite{Bre13}.
We have presented here several numerical examples that illustrate the
effectiveness of our basic approach and its competitiveness with weighted balanced
truncation. 

\bibliographystyle{plain}        
\bibliography{references}           



\end{document}